\documentclass[journal]{IEEEtran}
\usepackage{subfigure}
\usepackage{setspace}
\usepackage{amsmath}
\usepackage{amssymb}
\usepackage{amsfonts}
\usepackage{amscd}
\usepackage{mathrsfs}
\usepackage[final]{graphicx}
\usepackage{graphicx}
\usepackage{psfrag}
\usepackage{epsfig}
\usepackage{color}
\usepackage{url}
\usepackage{textcomp}
\usepackage{multirow}
\usepackage{threeparttable}
\input{epsf.sty}
\newtheorem{theorem}{Theorem}
\newtheorem{lemma}{Lemma}
\newtheorem{definition}{Definition}
\newtheorem{proposition}{Proposition}
\newtheorem{corollary}{Corollary}
\newtheorem{example}{Example}
\newtheorem{remark}{Remark}
\newtheorem{note}{Note}
\newtheorem{case}{Case}

\begin{document}

% paper title
\title{Fast-Decodable MIDO Codes with Large Coding Gain}
\vspace{1.00cm}
\author{K. Pavan Srinath and B. Sundar Rajan, {\it Senior Member, IEEE}
\thanks{K. Pavan Srinath is with Broadcom Communication Technologies Pvt. Ltd., Bangalore. This work was carried out when he was with the Department of Electrical Communication Engineering, Indian Institute of Science, Bangalore. Email: srinath.pavan@gmail.com.}
 \thanks {B. Sundar Rajan is with the Department of ECE, Indian Institute of Science, Bangalore - 560012. Email: bsrajan@ece.iisc.ernet.in.}
\thanks {Part of the content of this manuscript will appear in the proceedings of IEEE Int. Symp. Inf. Theory (ISIT 2013), Istanbul, Turkey, July 07-12, 2013.}
}

% make the title area
\maketitle
\vspace{-15mm}
\begin{abstract}
In this paper, a new method is proposed to obtain full-diversity, rate-$2$ (rate of 2 complex symbols per channel use) space-time block codes (STBCs) that are full-rate for multiple input, double output (MIDO) systems. Using this method, rate-$2$ STBCs for $4\times2$, $6 \times 2$, $8\times2$ and $12 \times 2$ systems are constructed and these STBCs are fast ML-decodable, have large coding gains and STBC-schemes consisting of these STBCs have a non-vanishing determinant (NVD) so that they are DMT-optimal for their respective MIDO systems. It is also shown that the Srinath-Rajan code [R. Vehkalahti, C. Hollanti, and F. Oggier, ``Fast-Decodable Asymmetric Space-Time Codes from Division Algebras,'' IEEE
Trans. Inf. Theory, Apr. 2012] for the $4\times2$ system, which has the lowest ML-decoding complexity among known rate-2 STBCs for the $4\times2$ MIDO system with a large coding gain for $4$-$/16$-QAM, has the same algebraic structure as the STBC constructed in this paper for the $4\times2$ system. This also settles in positive a previous conjecture that the STBC-scheme that is based on the Srinath-Rajan code has the NVD property and hence is DMT-optimal for the $4\times2$ system.
\end{abstract}

 \begin{keywords}
Cyclic division algebra, fast-decodability, Galois group, MIDO system, non-vanishing determinant, space-time block codes. 
 \end{keywords}

%%%%%%%%%%%%%%%%%%%%%%%%%%%%%%%%%%%%%%%%%%%%%%%%%%%%%%%%%%%%%%%%%%%%%%%%%%%%%%%%%%%%%%
\begin{table*}
\begin{center}
\begin{threeparttable}
\begin{tabular}{|c|c|c|c|c|} \cline{1-5}
\multirow{3}{*}{$\#$ Tx antennas} & \multirow{3}{*}{STBC $\mathcal{S}$}  & Constellation &\multirow{3}{*}{$\delta_{min}(\mathcal{S})$} & ML-decoding \\ 
& & (average energy $E$) & & complexity\\ 
& & & & (Worst case)\\ \hline \hline
 \multirow{10}{*}{$4$} & \multirow{2}{*}{$\mathcal{S}_{4\times 2}$} & \multirow{2}{*}{QAM} &   \multirow{2}{*}{$ \frac{1}{25E^4}$}  &  \multirow{2}{*}{$ \mathcal{O}\left(M^{4.5}\right)$ }\\ 
& & & & \\ \cline{2-5}
 & Punctured\tnote{{\bf\textdollar}}  & \multirow{2}{*}{QAM} & \multirow{2}{*}{$\frac{16}{1125E^4}$} & \multirow{2}{*}{$\mathcal{O}\left(M^{5.5}\right)$} \\ 
& perfect code \cite{ORBV} & & & \\ \cline{2-5}
& \multirow{2}{*}{$\mathcal{C}_1$ \cite[Sec. VIII-B]{roope}, \cite{spcom}} & \multirow{2}{*}{QAM} & \multirow{2}{*}{$ \frac{1}{25E^4}$} & \multirow{2}{*}{$ \mathcal{O}\left(M^{6.5}\right)$} \\
& & & & \\ \cline{2-5}
& \multirow{2}{*}{$A_4$ \cite[Sec. VIII-A]{roope}} & \multirow{2}{*}{QAM} & \multirow{2}{*}{Not Available} & \multirow{2}{*}{$ \mathcal{O}\left(M^{5.5}\right)$} \\
& & & & \\ \cline{2-5}
& \multirow{2}{*}{Punctured $\mathcal{C}_4$ \cite{pav_modified}} & \multirow{2}{*}{QAM} & \multirow{2}{*}{$ \frac{1}{16E^4}$} & \multirow{2}{*}{$ \mathcal{O}\left(M^{7}\right)$} \\
& & & & \\ \hline \hline
 \multirow{10}{*}{$6$} & \multirow{2}{*}{$\mathcal{S}_{6\times 2}$} & \multirow{2}{*}{HEX} &  \multirow{2}{*}{$\frac{1}{7^4E^6}$} & \multirow{2}{*}{$\mathcal{O}\left(M^{8.5}\right)$} \\ 
& & & & \\ \cline{2-5}
 & Punctured & \multirow{2}{*}{HEX} & \multirow{2}{*}{$\frac{1}{7^5E^6} \leq \delta_{min}\leq \frac{1}{7^4E^6}$} & \multirow{2}{*}{$\mathcal{O}\left(M^{11.5}\right)$} \\
 & perfect code \cite{ORBV} &  & & \\ \cline{2-5}
 & \multirow{2}{*}{Punctured $\mathcal{C}_6$ \cite{pav_modified}} & \multirow{2}{*}{HEX} & \multirow{2}{*}{$\frac{1}{(3E)^6}$} & \multirow{2}{*}{$\mathcal{O}\left(M^{11.5}\right)$} \\
 &  &  & & \\ \cline{2-5}
 & VHO-Code & \multirow{2}{*}{QAM} & \multirow{2}{*}{Not available\tnote{{\bf\pounds}}} & \multirow{2}{*}{ $ \mathcal{O}\left(M^{8.5}\right)$} \\
& \cite[Sec. X-C]{roope} & & & \\ \cline{2-5}
& VHO-Code & \multirow{2}{*}{QAM} & \multirow{2}{*}{Not available\tnote{{\bf\textyen}}} & \multirow{2}{*}{ $ \mathcal{O}\left(M^{7}\right)$} \\
& (Change of Basis) & & & \\ \hline \hline
 \multirow{6}{*}{$8$}& \multirow{2}{*}{$\mathcal{S}_{8\times 2}$} & \multirow{2}{*}{QAM} &  \multirow{2}{*}{$\frac{1}{25(15)^4E^8}$} & \multirow{2}{*}{$ \mathcal{O}\left(M^{9.5}\right)$} \\ 
& & & & \\ \cline{2-5}
  & Punctured  & \multirow{2}{*}{QAM} & \multirow{2}{*}{$\frac{1}{5^72^{16}E^8}$} & \multirow{2}{*}{$\mathcal{O}\left(M^{15.5}\right)$} \\
&perfect code \cite{new_per} & & & \\ \cline{2-5}
 & \multirow{2}{*}{VHO-Code \cite{roope}} & \multirow{2}{*}{QAM} &  \multirow{2}{*}{Not Available\tnote{{\bf\textyen}}} & \multirow{2}{*}{$ \mathcal{O}\left(M^{9.5}\right)$} \\ 
& & & & \\    \hline \hline
\multirow{4}{*}{$12$} & \multirow{2}{*}{VHO-Code \cite{roope}} & \multirow{2}{*}{QAM} &  \multirow{2}{*}{Not Available\tnote{{\bf\textyen}}} & \multirow{2}{*}{$ \mathcal{O}\left(M^{14.5}\right)$} \\ 
& & & & \\   \cline{2-5}
 & \multirow{2}{*}{$\mathcal{S}_{12\times 2}$} & \multirow{2}{*}{HEX} &  \multirow{2}{*}{$\delta_{min} \geq  \frac{1}{(14E)^{12}}$} & \multirow{2}{*}{$ \mathcal{O}\left(M^{17.5}\right)$} \\ 
& & & & \\ \hline
\end{tabular}
\begin{tablenotes}
\item[{\bf \textdollar}] Punctured STBCs for $n_r <n_t$ refer to rate-$n_r$ STBCs obtained from rate-$n_t$ STBCs (which transmit $n_t^2$ complex information symbols in $n_t$ channel uses) by restricting the number of complex information symbols transmitted to be only $n_tn_r$.
\item[{ \bf \pounds}] The exact minimum determinant of this STBC has not been calculated, but it has been shown that the STBC has the NVD property \cite{roope}.
\item[{ \bf \textyen}] These STBCs are not available explicitly in \cite{roope}. However, it is possible to construct such STBCs with the ML-decoding complexities shown in corresponding row.
\end{tablenotes}
\end{threeparttable}
\end{center}
\caption{Comparison of our STBCs with known best STBCs.}
\label{tab1}
\hrule
\end{table*}
%%%%%%%%%%%%%%%%%%%%%%%%%%%%%%%%%%%%%%%%%%%%%%%%%%%%%%%%%%%%%%%%%%%%%%%%%%%%%%%%%%%%%%%%%%%%%%%%%%%%%

\section{Introduction and Background}\label{sec_intro}
Space-time block coding \cite{TSC} has been continually evolving over the last decade. Beginning with the simple Alamouti code \cite{SMA} for 2 transmit antennas, the evolution of space-time coding theory has resulted in the development of sophisticated full-diversity codes from cyclic division algebras (CDAs) \cite{sethuraman}-\cite{new_per} for any number of transmit antennas. At one end are the rate-1 (see Definition \ref{code_rate}) STBCs that are multi-group decodable (see, for example, \cite{KhR}-\cite{4gp2} for a definition of multi-group decodable STBCs) and have a relatively low maximum likelihood (ML)-decoding complexity while at the other end are rate-$n_t$ (for $n_t$ transmit antennas) full-diversity STBCs obtained from CDAs which have a very high ML-decoding complexity. The usage of powerful tools from number theory has resulted in rate-$n_t$ (for $n_t$ transmit antennas) STBCs with high coding gains, and STBC-schemes (see Definition \ref{scheme}) employing these codes have a non-vanishing determinant (see Definition \ref{nvd_def}) so that they are diversity-multiplexing gain tradeoff (DMT)-optimal \cite{elia} for any number of receive antennas. Examples of such codes are the perfect codes \cite{ORBV}, \cite{new_per}.

Recent interest has been towards asymmetric MIMO systems where the number of receive antennas $n_r$ is less than the number of transmit antennas $n_t$. Such a scenario occurs, for example, in the downlink transmission from a base station to a mobile phone, and in digital video broadcasting (DVB) where communication is between a TV broadcasting station and a portable TV device (see, for example, \cite{dvb}). Of particular interest is the $4\times 2$ MIDO system for which a slew of rate-$2$ STBCs have been developed \cite{roope}, \cite{pav_rajan2}, \cite{BHV}-\cite{markin}, with the particular aim of allowing fast-decodability (see Definition \ref{fast_dec}), a term that was first coined in \cite{BHV}. Among these codes, those in \cite{roope}, \cite{FVC}-\cite{markin} have been shown to have a minimum determinant that is bounded away from zero irrespective of the size of the signal constellation and hence STBC-schemes that consist of these codes have the NVD property and are DMT-optimal for the $4\times2$ MIDO system \cite{pav_rajan_dmt}. A generalization of fast-decodable STBC construction for higher number of transmit antennas has been proposed in \cite{roope}. STBCs from nonassociative division algebras have also been proposed in \cite{unger}.

The best performing code for the $4\times2$ MIDO system is the Srinath-Rajan code \cite{pav_rajan2} which has the least ML-decoding complexity (of the order of $M^{4.5}$ for a square $M$-QAM) among comparable codes and the best known normalized minimum determinant (see Definition \ref{min_det_stbc}) for $4$-$/16$-QAM. However, this code was constructed using an ad hoc technique and had not been proven to have a non-vanishing determinant for arbitrary QAM constellations. In this paper, we propose a novel construction scheme to obtain rate-$2$ STBCs which have full-diversity, and STBC-schemes that employ these codes have the NVD property. We then explicitly construct such STBCs for $n_t\times2$ MIDO systems, $n_t=4,6,8,12$, and these codes are fast-decodable and have large normalized minimum determinants.

\subsection{Contributions and paper organization}
The contributions of this paper may be summarized as follows.
\begin{enumerate}
 \item We propose a novel algebraic method to construct rate-$2$ STBCs with full-diversity. A highlight of our mathematical framework is that it is a generalization of the frameworks of \cite{markin} and \cite{unger}. 
\item Using our construction methodology, we construct rate-$2$, fast-decodable STBCs for $4\times2$, $6\times2$, $8\times2$ and $12\times2$ MIDO systems. All these four STBCs have large normalized minimum determinants and fast-decodability (see Table \ref{tab1}). In addition, STBC-schemes that consist of these STBCs have the NVD property making them DMT-optimal for their respective MIDO systems.
\item We show that the Srinath-Rajan (SR) code \cite{pav_rajan2} has the same underlying algebraic structure as the STBC constructed in this paper for the $4\times2$ system. This way, we prove the conjecture that the STBC-scheme based on the SR-code has the NVD property.
\end{enumerate}

The paper is organized as follows. Section \ref{sec_system_model} gives the system model, relevant definitions and a brief overview of CDAs. Section \ref{sec_gen_scheme} builds the theory needed to obtain rate-$2$ STBCs, while Section \ref{sec_STC} deals with the construction of fast-decodable STBCs for $4\times2$, $6\times2$, $8\times2$ and $12\times2$ systems. The property of the constructed STBCs that allows fast-decodability is explained in Section \ref{sec_ml_comp}, and simulation results are given in Section \ref{sec_sim}. Concluding remarks constitute Section \ref{sec_discussion}.

\subsection*{Notation}
\noindent Throughout the paper, the following notation is employed. 
\begin{itemize}
 \item Bold, lowercase letters denote vectors, and bold, uppercase letters denote matrices.
 \item $\textbf{X}^{H}$, $\textbf{X}^{T}$, $det(\textbf{X})$, $tr(\textbf{X})$ and $\Vert \textbf{X} \Vert$ denote the conjugate transpose, the transpose, the determinant, the trace and the Frobenius norm of $\textbf{X}$, respectively. 
\item $\textrm{diag}[\textbf{A}_1,\textbf{A}_2,\cdots,\textbf{A}_n]$ denotes a block diagonal matrix with matrices $\textbf{A}_1$, $\textbf{A}_2$, $\cdots$, $\textbf{A}_n$ on its main diagonal blocks.
\item The real and the imaginary parts of a complex-valued vector $\textbf{x}$ are denoted by $\mathcal{R}e(\textbf{x})$ and $\mathcal{I}m(\textbf{x})$, respectively.
\item $\vert \mathcal{S}\vert$ denotes the cardinality of the set $\mathcal{S}$, and for a set $\mathcal{T}$ such that $\vert \mathcal{T}\cap \mathcal{S}\vert \neq 0 $, $\mathcal{S} \setminus  \mathcal{T}$ denotes the set of elements of $\mathcal{S}$ not in $\mathcal{T}$. 
\item $\textbf{I}$ and $\textbf{O}$ denote the identity and the null matrix of appropriate dimensions.
\item $\mathbb{E}(X)$ denotes the expectation of the random variable $X$. 
\item $\mathbb{R}$, $\mathbb{C}$ and $\mathbb{Q}$ denote the field of real, complex and rational numbers, respectively, and $\mathbb{Z}$ denotes the ring of rational integers. 
\item Unless used as an index, a subscript or a superscript, $i$ denotes $\sqrt{-1}$ and $\omega$ denotes the primitive third root of unity. 
\item For fields $K$ and $F$, $K/F$ denotes that $K$ is an extension of $F$ (hence, $K$ is an algebra over $F$) and $[K:F ] = m$ indicates that $K$ is a finite extension of $F$ of degree $m$.
\item $M_n(K)$ denotes the ring of $n\times n$ sized matrices with entries from a field $K$. 
\item$Gal( K/F)$ denotes the Galois group of $K/F$, i.e., the group of $F$-linear automorphisms of $K$. If $\sigma$ is any $F$-linear automorphism of $K$, $\langle \sigma \rangle$ denotes the cyclic group generated by $\sigma$. 
\item The elements $1$ and $0$ are understood to be the multiplicative identity and the additive identity, respectively, of the unit ring $\mathcal{R}$ in context.
\item $\textrm{im}(\Phi)$ denotes the image of the map $\Phi$.
\end{itemize}

\section{System Model and definitions}\label{sec_system_model}
We consider an $n_t$ transmit antenna, $n_r$ receive antenna MIMO system ($n_t\times n_r$ system) with perfect channel-state information available at the receiver (CSIR) alone. The channel is assumed to be quasi-static with Rayleigh fading. The system model is 
\begin{equation}\label{model}
 \textbf{Y} = \sqrt{\rho}\textbf{HS} + \textbf{N},
\end{equation}

\noindent where $\textbf{Y} \in \mathbb{C}^{n_r\times \textrm{T}}$ is the received signal matrix, $\textbf{S} \in \mathbb{C}^{n_t\times \textrm{T}}$ is the codeword matrix that is transmitted over a block of $\textrm{T}$ channel uses, $\textbf{H} \in \mathbb{C}^{n_r\times n_t}$ and $\textbf{N} \in \mathbb{C}^{n_r\times \textrm{T}}$ are respectively the channel matrix and the noise matrix with entries independently and identically distributed (i.i.d.) circularly symmetric complex Gaussian random variables with zero mean and unit variance. The average signal-to-noise ratio (SNR) at each receive antenna is denoted by $\rho$. It follows that 
\begin{equation}\label{energy_con}
 \mathbb{E}(\Vert \textbf{S} \Vert^2) = \textrm{T}.
\end{equation}

A space-time block code (STBC) $\mathcal{S}$ of block-length $\textrm{T}$ for an $n_t$ transmit antenna MIMO system is a finite set of complex matrices of size $n_t \times \textrm{T}$. Throughout the paper, we consider linear STBCs \cite{HaH} encoding symbols from a complex constellation $\mathcal{A}_q$ which is QAM or HEX. An $M$-PAM, $M$-QAM and $M$-HEX constellation, with $M = 2^a$, $a$ even and positive, are respectively given as  
\begin{eqnarray*}
 M\textrm{-PAM} & = & \{ -M+1,-M+3, -M+5,\cdots,M-1\},
\end{eqnarray*}
\begin{eqnarray*}
 M\textrm{-QAM} & = & \left\{ a + ib, a,b \in \sqrt{M}\textrm{-PAM}\right\},\\
 M\textrm{-HEX} & = & \left\{ a + \omega b, a,b \in \sqrt{M}\textrm{-PAM}\right\}.
\end{eqnarray*}
Assuming that $\mathcal{A}_q$ is $M$-QAM or $M$-HEX, the symbols $s_i$ encoded by the STBC are of the form $s_i \triangleq \bar{s}_i + \beta\check{s}_i$, with $\bar{s}_i, \check{s}_i \in \sqrt{M}\textrm{-PAM}$ and $\beta = i$ or $\omega$ depending on whether $\mathcal{A}_q$ is $M$-QAM or $M$-HEX, respectively. Therefore, the STBC $\mathcal{S}$ is of the form 
\begin{equation}\label{form_stbc}
 \mathcal{S} = \left\{ \left. \textbf{S}_i = \sum_ {j=1}^k \left(\bar{s}_{ij}\bar{\textbf{A}}_{j} + \check{s}_{ij}\check{\textbf{A}}_{j}\right)~\right \vert ~s_{ij} \in \mathcal{A}_q \right \}
\end{equation}
 where $\textbf{S}_i$, $i=1,2,\cdots,\vert \mathcal{A}_q \vert^k$, are the codeword matrices, and $\bar{\textbf{A}}_{j}$ and $\check{\textbf{A}}_{j}$ are complex weight matrices of the STBC. We assume that the average energy of $\mathcal{A}_q$ is $E$ units. Noting the symmetry of both $M$-QAM and $M$-HEX, we have $\mathbb{E}(\vert \bar{s}_{ij} \vert^2) = \mathbb{E}(\vert \check{s}_{ij} \vert^2) = E/2$, $\mathbb{E}(\bar{s}_{ij}\check{s}_{ij}) = 0$. So, the energy constraint in \eqref{energy_con} translates to $ E\sum_{i=1}^{k}tr\left(\bar{\textbf{A}}_{i}\bar{\textbf{A}}_{i}^H + \check{\textbf{A}}_{i}\check{\textbf{A}}_{i}^H \right) = 2\textrm{T}$. 
\begin{definition}\label{gen_mat}
({\it Generator matrix}) An STBC of the form
\begin{equation*}
 \mathcal{S} = \left\{ \textbf{S}_i = \sum_ {j=1}^k \left(\bar{s}_{ij}\bar{\textbf{A}}_{j} + \check{s}_{ij}\check{\textbf{A}}_{j}\right)~\vert~\bar{s}_{ij},\check{s}_{ij} \in \textrm{PAM} \right \},
\end{equation*}
is said to be a space-time lattice code \cite{roope}, and its generator matrix $\textbf{G}$ is given as 
\begin{equation*}
 \textbf{G} =  \left[\widetilde{vec(\bar{\textbf{A}}_{1})} ~ \widetilde{vec(\check{\textbf{A}}_{1})} \cdots \widetilde{vec(\bar{\textbf{A}}_{k})} ~ \widetilde{vec(\check{\textbf{A}}_{k})} \right],
\end{equation*}
where for a matrix $\textbf{X} = [\textbf{x}_1 ~\textbf{x}_2 \cdots ~\textbf{x}_m]$ with $\textbf{x}_i$ being column vectors, $vec(\textbf{X}) \triangleq [\textbf{x}_1^T ~\textbf{x}_2^T \cdots \textbf{x}_m^T ]^T$, and for a complex column vector $\textbf{x} = [x_1,x_2,\cdots,x_n]^T$, $\tilde{\textbf{x}} \triangleq [\mathcal{R}e(x_1),\mathcal{I}m(x_1),\cdots,\mathcal{R}e(x_n),\mathcal{I}m(x_n)]^T$.
\end{definition}
\begin{definition}\label{code_rate}
 ({\it STBC Rate}) The rate of an STBC is $\frac{Rank(\textbf{G})}{2\textrm{T}}$ complex symbols per channel use, where $\textbf{G}$ is the generator matrix of the STBC. The STBC is said to encode $Rank(\textbf{G})/2$ independent complex symbols.
\end{definition}

An STBC having a rate of $\min(n_t,n_r)$ complex symbols per channel use is said to be a {\it full-rate} STBC.

\begin{definition}\label{cubic_shaping}
({\it Cubic Shaping}) If the generator matrix $\textbf{G}$ has column orthogonality, the space-time lattice code is said to have cubic shaping \cite{roope}. 
\end{definition}

Among STBCs transmitting at the same rate in bits per channel use (the bit rate of $\mathcal{S}$ is $\frac{\log_2\vert \mathcal{S} \vert}{\textrm{T}}$ bits per channel use), the metric for comparison that decides their error performance is the normalized minimum determinant which is defined as follows.

\begin{definition}\label{min_det_stbc}({\it Normalized minimum determinant})
 For an STBC $\mathcal{S} = \{ \textbf{S}_i, i = 1,\cdots,\vert \mathcal{S}\vert\}$ that satisfies \eqref{energy_con}, the normalized minimum determinant $\delta_{min}(\mathcal{S})$ is defined as
\begin{equation}\label{min_det_eq}
\delta_{min}(\mathcal{S}) = \min_{\textbf{S}_i,\textbf{S}_j \in \mathcal{S}, i\neq j} \left \vert det\left(\textbf{S}_i - \textbf{S}_j\right) \right\vert^2.  
\end{equation}
\end{definition}
For full-diversity STBCs, $\delta_{min}(\mathcal{S})$ defines the coding gain \cite{TSC}, with the coding gain given by $\delta_{min}(\mathcal{S})^{\frac{1}{n_t}}$. Between two competing STBCs, the one with the larger normalized minimum determinant is expected to have a better error performance. 
\begin{note}\label{note1}
When the average energy of transmission in each time slot is uniform, then the energy constraint given by \eqref{energy_con} implies that $\mathbb{E}(\Vert \textbf{s}_i \Vert^2) =1 $, $\forall i = 1,\cdots,\textrm{T}$, where $\textbf{s}_i$ denotes the $i^{th}$ column of a codeword matrix.
\end{note}

\begin{definition}\label{scheme} ({\it STBC-scheme} \cite{tse}) An STBC-scheme $\mathcal{X}_{scheme}$ is defined as a family of STBCs indexed by $\rho$, each STBC of block length $\textrm{T}$ so that $\mathcal{X}_{scheme} = \{ \mathcal{S}(\rho)\}$, where the STBC $\mathcal{S}(\rho)$ corresponds to an average signal-to-noise ratio of $\rho$ at each receive antenna.  
\end{definition}

For STBC-schemes that consist of linear STBCs employing complex lattice constellations, the weight matrices define the STBC-scheme. The weight matrices are fixed and the size and average energy of the signal constellation are allowed to vary in accordance with $\rho$. Associated with such linear STBC-schemes is the notion of non-vanishing determinant (NVD).    

\begin{definition}\label{nvd_def}({\it Non-vanishing determinant} \cite{BRV}) A linear STBC-scheme $\mathcal{X}_{scheme}$, whose STBCs are defined by weight matrices $\{ \bar{\textbf{A}}_{i}, \check{\textbf{A}}_{i},i=1,\cdots,k \}$ and employ complex constellations that are finite subsets of an infinite complex lattice $\mathcal{A}_L$, is said to have the non-vanishing determinant (NVD) property if $\mathcal{S}_{\infty}  \triangleq \left\{ \sum_{i=1}^k \left(\bar{s}_{i}\bar{\textbf{A}}_{i} + \check{s}_{i}\check{\textbf{A}}_{iQ}\right) \vert s_{i} \in \mathcal{A}_L \right\}$ is such that 
\begin{equation*}
 \inf_{\textbf{S} \in \mathcal{S}_{\infty} \setminus \textbf{O}} \left\{ \vert det(\textbf{S})  \vert^2\right\}  > 0.
\end{equation*}
\end{definition}

With respect to ML-decoding, if the STBC transmits $k$ complex symbols in $\textrm{T}$ channel uses where the symbols are encoded from a suitable complex constellation of size $M$, an exhaustive search requires performing $\mathcal{O}\left(M^k\right)$ operations ($\mathcal{O}()$ stands for ``big O of'') because the $k$ symbols have to be jointly evaluated. However, some STBCs allow fast-decodability which is defined as follows. 

\begin{definition}\label{fast_dec}({\it Fast-decodable STBC} \cite{BHV}) Consider an STBC encoding $k$ complex information symbols from a complex constellation of size $M$. If the ML-decoding of this STBC by an exhaustive search involves performing only $\mathcal{O}\left(M^p\right)$ computations, $p < k$, the STBC is said to be fast-decodable.  
\end{definition}

For more on fast-decodability, one can refer to \cite{pav_rajan2}, \cite{BHV}.

\subsection{Cyclic Division Algebras}
A cyclic division algebra (CDA) $\mathcal{A}$ of degree $n$ over a number field $F$ is a vector space over $F$ of dimension $n^2$. The centre of $\mathcal{A}$, denoted by $ Z(\mathcal{A}) $ and defined as 
\begin{equation*}
 Z(\mathcal{A}) = \{A \in \mathcal{A} ~\vert ~ AB = BA, ~ \forall B \in \mathcal{A} \},
\end{equation*}
 is the field $F$ itself, and there exists a maximal subfield $K$ of $\mathcal{A}$ such that $K$ is a Galois extension of degree $n$ over $F$ with a cyclic Galois group generated by a cyclic generator $\tau$. $\mathcal{A}$ is a right vector space over $K$ and can be expressed as $
 \mathcal{A} = K \oplus \textbf{i}K \oplus \textbf{i}^2K \oplus \cdots \oplus \textbf{i}^{n-1}K$, where $ a\textbf{i}  =  \textbf{i}\tau(a)$, $\forall a \in K$, $\textbf{i}^n  =  \gamma$ for some $\gamma \in F^{\times} = F\setminus \{0\}$ such that the norm $N_{K/F}(a) = \prod_{i=0}^{n-1}\tau^i(a)$ of any element $a \in K$ satisfies \cite[Proposition 2.4.5]{hollanti_thesis} 
\begin{equation}\label{non_norm_condition}
 N_{K/F}(a) \neq \gamma^{p} 
\end{equation}
for any divisor (in $\mathbb{Z}$) $p$ of $n$ with $1 \leq p < n$. The CDA $\mathcal{A}$ is denoted by $(K/F, \tau, \gamma)$. $\mathcal{A}$ has a matrix representation in $M_n(K)$. This means that there exists an injective ring homomorphism from $\mathcal{A}$ to the matrix ring $M_n(K)$, described as follows. The map 
\begin{equation}\label{left_regular}
 \lambda_a : d \mapsto ad , ~~~\forall d \in \mathcal{A},
\end{equation}
for $a \in \mathcal{A}$ is called the left regular map, and there exists an injective ring homomorphism (specifically, an isomorphism) $\Phi$ from $\mathcal{A}$ to the ring $ \Lambda = \{\lambda_a ~ \vert ~ a \in \mathcal{A} \}$, given by 
 \begin{eqnarray*}
  \Phi : \mathcal{A} & \longrightarrow & \Lambda \\
      a & \mapsto & \lambda_a.
 \end{eqnarray*}
Since every nonzero element of $\mathcal{A}$ is invertible, $\textrm{im}(\Phi)$ with the exception of the zero map consists of invertible maps from $\mathcal{A}$ to itself. Each $\lambda_a \in \Lambda$ is a $K$-linear transformation of the right $K$-vector space $\mathcal{A}$, and hence is associated with a matrix in $M_n(K)$. In particular, $\lambda_a$, where $a = a_0+{\bf i}a_1+\cdots+{\bf i}^{n-1}a_{n-1}  \in \mathcal{A}$ with $a_i \in K$, is associated with the matrix ${\bf F}_a \in M_n(K)$ which is the matrix representation of $\lambda_a$ and is given as 
\begin{equation}\label{form_div}
 {\bf F}_a = \left[ \begin{array}{cccc}
                      a_0 & \gamma\tau(a_{n-1}) & \cdots & \gamma\tau^{n-1}(a_{1})\\
                      a_1 & \tau(a_0)  & \cdots & \gamma\tau^{n-1}(a_2)\\
                      \vdots &  \vdots &  & \vdots\\
                      a_{n-1} & \tau(a_{n-2}) & \cdots & \tau^{n-1}(a_0)\\                      
                     \end{array} \right].
\end{equation}
It follows that 
\begin{eqnarray*}
 \lambda_a(x) = ax & = &[ 1,~{\bf i},~\cdots,~{\bf i}^{n-1}]{\bf F}_a{\bf x} 
\end{eqnarray*}
where $x = x_0+{\bf i}x_1+\cdots+{\bf i}^{n-1}x_{n-1} \in \mathcal{A}$ and ${\bf x} = [x_0, \cdots,x_{n-1}]^T \in K^{n\times1}$. The ring homomorphism (it can be easily checked that it is indeed a ring homomorphism)
\begin{eqnarray*}
\Psi : \Lambda & \longrightarrow & M_n(K) \\
          \lambda_a & \mapsto & F_a 
\end{eqnarray*}
is injective, and so is the ring homomorphism 
\begin{eqnarray*}
\Psi \circ \Phi : \mathcal{A} & \longrightarrow & M_n(K) \\
                     a & \mapsto & \textbf{F}_a .
\end{eqnarray*}
Since $\mathcal{A}$ is a division algebra, so is $\textrm{im}(\Psi \circ \Phi)$, and hence, every nonzero matrix of the form shown in \eqref{form_div} is invertible. It is known that \cite{jacobson} $
 det({\bf F}_a) \in F^{\times}$, ${\bf F}_a \neq {\bf O}$. For more on CDAs, one can refer to \cite{sethuraman}, \cite{jacobson}, and references therein. 

\subsection{STBCs from CDA}\label{stc_cda}

 In this section, we review some known techniques to obtain full-diversity STBC-schemes with a non-vanishing determinant and large coding gain. For the purpose of space-time coding, the signal constellation is generally $M$-QAM, $M$-HEX or $M$-PAM which are finite subsets of $\mathbb{Z}[i]$, $\mathbb{Z}[\omega]$ and $\mathbb{Z}$, respectively. So, $F$ is naturally chosen to be $\mathbb{Q}(i)$, $\mathbb{Q}(\omega)$, or simply $\mathbb{Q}$ for which the ring of integers are respectively $\mathbb{Z}[i]$, $\mathbb{Z}[\omega]$ and $\mathbb{Z}$. By $\mathcal{O}_F$ and $\mathcal{O}_K$, we denote the ring of integers of $F$ and $K$ respectively. When $[K:F]=n$, $K$ has an $F$-basis of cardinality $n$. Similarly, $\mathcal{O}_K$ has an $\mathcal{O}_F$-basis of cardinality $n$. An $F$-basis $\{\theta_i, i=1,2,\cdots,n \vert \theta_i \in \mathcal{O}_K \}$ is chosen and the $a_i \in K$ in \eqref{form_div} are expressed as linear combinations of elements of this basis over $\mathcal{O}_F$. The STBC which encodes symbols from a complex constellation $\mathcal{A}_q$ ($M$-QAM, $M$-HEX) is given by $\mathcal{S} = \{ \textbf{S}_i, i=1,\cdots,\vert \mathcal{S} \vert\}$, where the codewords $\textbf{S}_i$ have the form shown in \eqref{form_div} with $a_i = \sum_{j=1}^{n}s_{ij}\theta_j$, $s_{ij} \in \mathcal{A}_q \subset \mathcal{O}_{F}$ with $\mathcal{O}_{F} = \mathbb{Z}[i], \mathbb{Z}[\omega]$ or $\mathbb{Z}$. A codeword matrix of STBCs from CDA has $n_t$ layers \cite{ORBV}, with the $(i+1)^{th}$ layer transmitting the vector $\textbf{D}_{i}[a_{i}, \tau(a_{i}), \cdots, \tau^{n-1}(a_{i}) ]^T$, $i =0,\cdots,n_t-1$, where
\begin{equation*}
 \textbf{D}_{i} \triangleq \textrm{diag}[\underbrace{1,\cdots,1}_{n_t-i \textrm{ times}}, \underbrace{\gamma,\cdots,\gamma}_{i\textrm{ times}}].
\end{equation*}
The $F$-basis $\{\theta_i, i=1,2,\cdots,n \vert \theta_i \in \mathcal{O}_K\}$ is generally chosen such that the matrix 
\begin{equation}\label{integral_basis}
 \textbf{R} = \left[ \begin{array}{cccc}
                      \theta_1 & \theta_2 &  \cdots & \theta_n\\
                      \tau(\theta_1) & \tau(\theta_2) &  \cdots & \tau(\theta_n)\\
                      \vdots & \vdots & \vdots  & \vdots \\
                      \tau^{n-1}(\theta_1) & \tau^{n-1}(\theta_2) &  \cdots & \tau^{n-1}(\theta_n)\\                     
                     \end{array} \right]
\end{equation}
is scaled unitary, i.e., $\textbf{RR}^H = \lambda \textbf{I}$ for some $\lambda \in \mathbb{R}$ (the scalar $1/\sqrt{\lambda}$ is the normalizing factor for $\textbf{R}$ to ensure that the energy constraint is satisfied). Further $\gamma$ is generally chosen such that $\vert \gamma \vert^2 =1 $. The perfect codes, which employ these techniques, have among the largest known coding gains in their comparable class.
\begin{note}
 In literature, certain spherically-shaped codes built from non-orthogonal maximal orders \cite{VHLR}-\cite{HL} and shaping lattices \cite{KC} have larger coding gains than the perfect codes. Notably, the Golden$+$ code \cite{VHLR}, a spherically-shaped STBC, has a better coding gain than the well-known Golden code for 2 transmit antennas which is a linear STBC with cubic shaping. However, spherically-shaped codes are not sphere-decodable, since they are not linear STBCs.
\end{note}

\section{Mathematical Framework}\label{sec_gen_scheme}
Let $F$ and $L$ be two distinct number fields and $K$ a Galois extension of both $F$ and $L$ such that 
\begin{enumerate}
 \item $Gal(K/F) = \langle \sigma \rangle$ with $\vert Gal(K/F) \vert = m$,
\item $Gal(K/L) = \langle \tau \rangle$ with $\vert Gal(K/L) \vert = n$,
\item $\sigma$ and $\tau$ commute, i.e., $\sigma\tau(a) = \tau\sigma(a)$, $\forall a \in K$.
\end{enumerate}
Let $\mathcal{A} = (K/F, \sigma, \gamma)$ be a cyclic division algebra of degree $m$ over $F$ with $\{1, \textbf{j}, \textbf{j}^2, \cdots, \textbf{j}^{m-1}\}$ being its basis as an $m$-dimensional right vector space over $K$.

We consider a non-commutative ring $\mathcal{M}_{\mathcal{A}}$ which is an $n$-dimensional bimodule over $\mathcal{A}$ (i.e., both a left $\mathcal{A}$-module and a right $\mathcal{A}$-module), but we will treat $\mathcal{M}_{\mathcal{A}}$ as a right $\mathcal{A}$-module in this paper. The structure of $\mathcal{M}_{\mathcal{A}}$ is as follows. We denote the elements of the basis of $\mathcal{M}_{\mathcal{A}}$ over $\mathcal{A}$ by $1$, $\textbf{i}$, $\textbf{i}^2$, $\cdots$, $\textbf{i}^{n-1}$. The elements of $\mathcal{M}_{\mathcal{A}}$ are of the form $A_0+\textbf{i}A_1 + \cdots +\textbf{i}^{n-1}A_{n-1}$ with $A_i \in \mathcal{A}$ and 
\begin{eqnarray}\label{d01}
 A\textbf{i} & = & \textbf{i}\Upsilon(A), ~~ \forall A \in \mathcal{A},\\ \label{d02}
 \textbf{i}^n & = & \gamma_{{}_M} ~~ \textrm{for some}~\gamma_{{}_M} \in \mathcal{A}
\end{eqnarray}
where 
\begin{equation}\label{d1}
 \Upsilon(A) \triangleq \tau(a_0) + \textbf{j} \tau(a_1) + \cdots + \textbf{j}^{m-1} \tau(a_{m-1})
\end{equation}
for $A = a_0 + \textbf{j} a_1 + \cdots + \textbf{j}^{m-1} a_{m-1}$, $a_1,\cdots, a_{m-1} \in K$. We further assume that $\gamma \in L$ so that $\tau(\gamma) = \gamma$. With this assumption and the fact that $\sigma$ and $\tau$ commute, we have
\begin{equation}\label{d2}
 \Upsilon(A)\Upsilon(B) = \Upsilon(AB), ~~A,B \in \mathcal{A}.
\end{equation}

\noindent Now, forcing the relation $\textbf{i}^{a}\textbf{i}^b = \textbf{i}^{a+b}$ for positive integral values of $a$ and $b$, \eqref{d02} implies that $\gamma_{{}_M}\textbf{i} = \textbf{i}\gamma_{{}_M}$ so that $\gamma_{{}_M}$ is invariant under $\Upsilon$. Hence, we require $\gamma_{{}_M}$ to be of the form $a_0 + \textbf{j}a_1+\cdots+\textbf{j}^{m-1}a_{m-1}$, $a_i \in L$, $i=0,1,\cdots,m-1$. In this paper, we only consider the case where $\gamma_{{}_M} \in L \subset K$. 
\begin{example}
 Consider $\mathcal{A}$ to be $(\mathbb{Q}(i,\sqrt{2})/\mathbb{Q}(\sqrt{2}),\sigma:i\mapsto -i, -1)$ which is known to be a division algebra and is a subalgebra of Hamilton's quaternions. Next consider the Galois extension $\mathbb{Q}(i,\sqrt{2})/\mathbb{Q}(i)$ whose Galois group is $\{1, \tau\}$ with $\tau:\sqrt{2} \to -\sqrt{2}$. Now, $\mathcal{M}_\mathcal{A}=\{A_0 + \textbf{i}A_1 \vert A_0,A_1 \in \mathcal{A}, \textbf{i}^2 = i\}$. If $A = a_0 + \sqrt{2}a_1 + i(a_2+\sqrt{2}a_3) + \textbf{j}\left(b_0 + \sqrt{2}b_1 + i(b_2+\sqrt{2}b_3)\right)$ with $a_i,b_i \in \mathbb{Q}$, then $\Upsilon(A) = a_0 - \sqrt{2}a_1 + i(a_2-\sqrt{2}a_3) + \textbf{j}\left(b_0- \sqrt{2}b_1 + i(b_2-\sqrt{2}b_3)\right)$.
\end{example}

In $\mathcal{M}_{\mathcal{A}}$, we seek conditions under which every element of the form $A_0 + \textbf{i}A_1$ has a unique right inverse, i.e., for every element of the form $A_0 + \textbf{i}A_1$, $A_0,A_1 \in \mathcal{A}$, there exists a {\it unique} element $B \in \mathcal{M}_{\mathcal{A}}$ such that $(A_0 + \textbf{i}A_1)B = 1$. Towards this end, we make use of the following lemma.
\begin{lemma} \label{lem1}
 A nonzero element $A$ of $\mathcal{M}_{\mathcal{A}}$, when it has a right inverse, has a unique right inverse if and only if it is not a left zero divisor, i.e., there exists no nonzero element $B \in \mathcal{M}_{\mathcal{A}}$ such that $AB = 0$. 
\end{lemma}
\begin{proof}
 If $A$ is not a left zero divisor, the uniqueness of the inverse follows, for if $AB = 1$ and $AB^\prime = 1$, then $A(B-B^\prime) = 0 \Rightarrow$ $B = B^\prime$. Conversely, if $A$ has a unique right inverse, it is not a left zero divisor, for if $AB = 1$ and $AC = 0$ for some $C \in \mathcal{A}$, then $A(B-C) = 1 \Rightarrow C = 0$.
\end{proof}

In the following theorem which is a generalization of \cite[Lemma 7]{markin}, we establish conditions under which each element of $\mathcal{M}_{\mathcal{A}}$ of the form $A_0 + \textbf{i}A_1$, $A_0, A_1 \in \mathcal{A}$, has a unique right inverse.

\begin{theorem}\label{thm1}
 Every nonzero element of $\mathcal{M}_{\mathcal{A}}$ of the form $A_0 + \textbf{i}A_1$, $A_0, A_1 \in \mathcal{A}$, has a unique right inverse if and only if
\begin{equation}\label{condt1}
 C\Upsilon(C)\Upsilon^2(C)\cdots\Upsilon^{n-1}(C) \neq \gamma_{{}_M} ~\textrm{for every} ~C \in \mathcal{A}.
\end{equation}
 \end{theorem}

The proof of Theorem \ref{thm1} is given in Appendix \ref{app_1}.
\begin{remark}
The condition in \eqref{condt1} is also necessary and sufficient for any element of the form $\textbf{i}^kA_0 + \textbf{i}^{l}A_1$, $A_0,A_1 \in \mathcal{A}$, $0 \leq k < l \leq n-1$, to have a unique right inverse. The proof is on similar lines to the proof of Theorem \ref{thm1}. 
\end{remark}
 
%Note that the requirement in \eqref{condt1} implies that $ c\tau(c)\tau^2(c)\cdots\tau^{n-1}(c) \neq \gamma_{{}_M} ~\textrm{for every} ~c \in K$, since $K \subset \mathcal{A}$. 
%In addition, if $ c\tau(c)\tau^2(c)\cdots\tau^{n-1}(c) \neq \gamma_{{}_M}^{t} ~\textrm{for every} ~c \in K$, $t=1,\cdots,n-1$, then, together with the rules specified in \eqref{d01}-\eqref{d1}, it is clear that $\mathcal{M}_{\mathcal{A}}$ contains another cyclic division algebra $\mathcal{A}_1 = (K/L,\tau,\gamma_{{}_M})$ whose maximal subfield is $K$, centre $L$ and basis (as a right vector space over $K$) $\{1, \textbf{i}, \cdots, \textbf{i}^{n-1}\}$ with $ a\textbf{i}  =  \textbf{i}\tau(a)$, $\forall a \in K$, $\textbf{i}^n  = \gamma_{{}_M}$ for some $\gamma_{{}_M} \in L^{\times}$. 
It is to be noted from \eqref{d01} that $\textbf{ij}= \textbf{ji}$. So, we have the following possibilities for $\gamma_{{}_M}$ and $\gamma$ (recall that $\gamma = \textbf{j}^m \in F^\times \cap L$).

\begin{case}\label{case_1}
$\gamma_{{}_M} \in L \cap F$. In this case $\mathcal{M}_\mathcal{A}$ is an associative algebra over $L \cap F$. 
\end{case}

\begin{case}\label{case_2}
 $\gamma_{{}_M}  \in  L \setminus F$. In this case, $\mathcal{M}_\mathcal{A}$ is never an associative algebra over $L \cap F$ and hence does not have a matrix representation, for if $\mathcal{M}_\mathcal{A}$ is an associative algebra over $L \cap F$ with $\gamma_{{}_M} \notin L \cap F$, then we have $\textbf{ji}^n = \textbf{i}^n\textbf{j}$ due to commutativity of $\textbf{i}$ and $\textbf{j}$, but  $(\textbf{ji})\textbf{i}^{n-1} = \textbf{j}(\textbf{ii}^{n-1}) = \textbf{ji}^n = \textbf{j}\gamma_{{}_M} \neq \gamma_{{}_M}\textbf{j} = \textbf{i}^n\textbf{j}$, leading to a contradiction.
\end{case}
 
In this paper, we consider the case $\gamma_{{}_M} \in L \setminus F$, $\gamma \in L \cap F$. Even though $\mathcal{M}_\mathcal{A}$ is now nonassociative and does not have a matrix representation, we still can make use of Theorem \ref{thm1} to obtain invertible matrices, which are desirable from the point of view of constructing full-diversity STBCs. In this direction, we arrive at the following result.

\begin{lemma}\label{lem2}
If every element $C$ of a cyclic division algebra $\mathcal{A}$ of degree $n$ is such that $C\Upsilon(C)\Upsilon^2(C)\cdots\Upsilon^{n-1}(C) \neq \gamma_{{}_M}$ where $\gamma_{{}_M}$ is some nonzero field element of $\mathcal{A}$, then the matrix  $\textbf{C}\Upsilon(\textbf{C})\Upsilon^2(\textbf{C})\cdots\Upsilon^{n-1}(\textbf{C}) - \gamma_{{}_M}\textbf{I}$ is invertible, where $\textbf{C}$, $\Upsilon(\textbf{C})$, $\cdots$, $\Upsilon^{n-2}(\textbf{C})$ and $\Upsilon^{n-1}(\textbf{C}) \in K^{m\times m}$ are respectively the matrix representations\footnote{Throughout the paper we denote by $\Upsilon(\textbf{C})$ the matrix obtained by applying $\tau$ to each entry of $\textbf{C}$, and for the special case of $\textbf{C}$ being the matrix representation of $C \in \mathcal{A}$, $\Upsilon(\textbf{C})$ happens to be the matrix representation of $\Upsilon(C)$.} of $C$, $\Upsilon(C)$, $\cdots$, $\Upsilon^{n-2}(C)$ and $\Upsilon^{n-1}(C)$ in $M_m(K)$.
\end{lemma}

The proof of Lemma \ref{lem2} is given in Appendix \ref{app_2}. We make use of Lemma \ref{lem2} to obtain the following result. 

\begin{theorem}\label{Thm2}
 Let $\mathcal{M}_{\mathcal{A}}$ be such that any element of the form $A_0+\textbf{i}A_1$ has a unique right inverse, and let $\textbf{A}_0$ and $\textbf{A}_1$ be matrix representations of $A_0$ and $A_1$, respectively, in $M_{m}(K)$. Consider the matrix 
\begin{equation}\label{M_form1}
 \textbf{M} = \left[ \begin{array}{cccc}
                      \textbf{A}_0 & \textbf{O} &  \cdots & \gamma_{{}_M}\Upsilon^{n-1}(\textbf{A}_1)\\
                      \textbf{A}_1 & \Upsilon(\textbf{A}_0) & \cdots & \textbf{O}\\
                      \textbf{O} & \Upsilon(\textbf{A}_1) & \cdots & \textbf{O}\\
                      \textbf{O} & \textbf{O}  & \cdots & \textbf{O}\\
                      \vdots & \vdots & \ddots & \textbf{O}\\
                      \textbf{O} & \textbf{O} & \cdots & \Upsilon^{n-1}(\textbf{A}_0)\\                      
                     \end{array} \right].
\end{equation}
Then,
\begin{enumerate}
 \item $\textbf{M}$ is invertible.
\item $det(\textbf{M}) \in L$.
\end{enumerate}
 
\end{theorem}

The proof of Theorem \ref{Thm2} is given in Appendix \ref{app_3}.

%%%%%%%%%%%%%%%%%%%%%%%%%%%%%%%%%%%%%%%%%%%%%%%%%%%%%%%%%%%%%%%%%%%%%%%%%%%%%%%%%%%%%%%%%%%%%%%%%%
\begin{table*}
\begin{equation}\label{A_str}
 \textbf{A}_k = \left[ \begin{array}{cccccc}
                      \sum_{i=1}^{n}s_{ki}\theta_i & \gamma \sigma \left( \sum_{i=1}^{n}s_{k(i+nm-n)}\theta_i\right) & \cdots  & \gamma\sigma^{m-1} \left( \sum_{i=1}^{n}s_{k(i+n)}\theta_i\right) \\
                      \sum_{i=1}^{n}s_{k(i+n)}\theta_i &  \sigma \left( \sum_{i=1}^{n}s_{ki}\theta_i\right) & \cdots  & \gamma\sigma^{m-1} \left( \sum_{i=1}^{n}s_{k(i+2n)}\theta_i\right) \\
                       \vdots & \vdots & \vdots & \vdots \\
                      \sum_{i=1}^{n}s_{k(i+nm-n)}\theta_i & \sigma \left( \sum_{i=1}^{n}s_{k(i+nm-2n)}\theta_i\right) & \cdots  & \sigma^{m-1} \left( \sum_{i=1}^{n}s_{ki}\theta_i\right) \\

                     \end{array} \right], ~~~ k =0 ,1.
\end{equation}
\hrule
\end{table*}
%%%%%%%%%%%%%%%%%%%%%%%%%%%%%%%%%%%%%%%%%%%%%%%%%%%%%%%%%%%%%%%%%%%%%%%%%%%%%%%%%%%%%%%%%%%%%%%%%%%

\begin{corollary}\label{cor1}
 If all the elements of $\textbf{M}$ are from $\mathcal{O}_{K}$, the ring of integers of $K$, then $det(\textbf{M}) \in L \cap \mathcal{O}_K = \mathcal{O}_{L}$. 
\end{corollary}

\subsection{Special Case: $n=2$}\label{subsec_a1}
For the case where $n = 2$ so that $\mathcal{M}_{\mathcal{A}} = \{A_0 + \textbf{i}A_1 ~\vert ~A_0,A_1 \in \mathcal{A}\}$, when $C\Upsilon(C) \neq \gamma_{{}_M}$, $\forall C\in\mathcal{A}$, from Theorem \ref{thm1}, there are no left zero divisors in $\mathcal{M}_\mathcal{A}$. It follows that there do not exist any right zero divisors in $\mathcal{M}_{\mathcal{A}}$. So, there are no zero divisors in $\mathcal{M}_{\mathcal{A}}$, i.e, if $AB = 0$ for $A , B \in \mathcal{M}_\mathcal{A}$, then either $A =0$ or $B =0$. It is well known that a nonassociative algebra is division if and only if it contains no zero divisors \cite[Page 15-16]{RDS}. So, $\mathcal{M}_\mathcal{A}$ is a nonassociative division algebra. However, the map   
\begin{eqnarray*}
 \Phi : \mathcal{M}_{\mathcal{A}} & \longrightarrow & M_{2m}(K) \\
        A_0 +\textbf{i}A_1 & \mapsto & \left[ \begin{array}{rr}
                          \textbf{A}_0& \gamma_{{}_M}\Upsilon(\textbf{A}_1) \\
                          \textbf{A}_1 & \Upsilon(\textbf{A}_0) \\
                         \end{array}\right]
\end{eqnarray*}
is not to be confused to be an injective homomorphism, even though $\textrm{im}(\Phi)$ excepting the zero matrix consists of invertible matrices. This is because $\mathcal{M}_{\mathcal{A}}$ is not an associative algebra. However, it is natural to wonder whether $\textrm{im}(\Phi)$ is obtainable as a matrix representation of some associative division algebra. We establish as follows that this is not the case. 

Firstly, $\textrm{im}(\Phi)$ is not an algebra since is not closed under multiplication. To see this, let $\textbf{A}$, $\textbf{B}$, $\textbf{C}$, and $\textbf{D}$ be matrix representations (in $M_m(K)$) of elements $A,B,C,D \in \mathcal{A}$. Then,
\begin{eqnarray*}
 &&\left[ \begin{array}{rr}
                          \textbf{A}& \gamma_{{}_M}\Upsilon(\textbf{B}) \\
                          \textbf{B} & \Upsilon(\textbf{A}) \\
                         \end{array}\right]\left[ \begin{array}{rr}
                          \textbf{C}& \gamma_{{}_M}\Upsilon(\textbf{D}) \\
                          \textbf{D} & \Upsilon(\textbf{C}) \\
                         \end{array}\right] = \\
 &&\left[ \begin{array}{rr}
                          \textbf{AC} +  \gamma_{{}_M}\Upsilon(\textbf{B})\textbf{D} &\gamma_{{}_M}\left(\Upsilon(\textbf{BC}) +  \textbf{A}\Upsilon(\textbf{D}) \right) \\
                          \textbf{BC} +  \Upsilon(\textbf{A})\textbf{D}  & \Upsilon(\textbf{AC}) +  \gamma_{{}_M}\textbf{B}\Upsilon(\textbf{D}) \\
                         \end{array}\right]
\end{eqnarray*}
which does not belong to $\textrm{im}(\Phi)$ since $\gamma_{{}_M}\textbf{I}$ is not the matrix representation (in $M_m(K)$) of $\gamma_{{}_M}$ as an element of $\mathcal{A}$. The matrix representation of $\gamma_{{}_M}$ in $M_m(K)$ is $\Gamma_M = \textrm{diag}[\gamma_{{}_M}, \sigma(\gamma_{{}_M}),\cdots,\sigma^{m-1}(\gamma_{{}_M}) ]\neq \gamma_{{}_M}\textbf{I}$ since $\gamma_{{}_M} \notin F$. 

Secondly, $\textrm{im}(\Phi)$ is not a subset of some division algebra. To see this, let
\begin{equation*}
 \textbf{E} = \left[\begin{array}{cc}
                     \textbf{O} & \gamma_{{}_M}\textbf{I} \\
                     \textbf{I} & \textbf{O}\\
                    \end{array}\right], ~~ \textbf{F} = \left[\begin{array}{cc}
                     \Gamma_M & \textbf{O} \\
                     \textbf{O} & \Gamma_M\\
                    \end{array}\right]
\end{equation*}
where $\Gamma_M$ is the matrix representation of $\gamma_{{}_M}$ in $M_m(K)$. Clearly, $\textbf{E},\textbf{F} \in \textrm{im}(\Phi)$, but $\textbf{E}^2 - \textbf{F}$ is nonzero and not invertible, and so $\textrm{im}(\Phi)$ cannot be a subset of a division algebra.

The mathematical frameworks considered in \cite{markin} and \cite{unger} to obtain full-diversity STBCs are special cases of our framework, and are briefly described in the following subsections.
 
\subsection{The full-diversity STBC construction technique of \cite{markin}}\label{subsec_a}
 In \cite{markin}, a CDA $\mathcal{A} = (K/F,\sigma,\gamma)$ ($m$ even), an Galois $L$-automorphism $\tau$ such that $\tau^2 = 1$, (with the possibility of $\tau$ belonging to $\langle \sigma \rangle$), and an element $\gamma_{{}_M} \in F$ such that $\tau(\gamma_{{}_M}) = \gamma_{{}_M}$, $A\Upsilon(A) \neq \gamma_{{}_M}$, $\forall A \in \mathcal{A}$ ($\Upsilon$ as defined in \eqref{d1}), are used to obtain full-diversity STBCs with fast-decodability for $2m\times2$ MIDO systems. The STBCs are obtained using the following map
\begin{eqnarray*}
 \Phi_{\gamma_{{}_M}} :  \mathcal{A} \times \mathcal{A} & \longrightarrow & M_{2m}(K) \\
  (A,B)& \mapsto &\left[ \begin{array}{rr}
                          \textbf{A} & \gamma_{{}_M}\Upsilon(\textbf{B}) \\
                          \textbf{B} & \Upsilon(\textbf{A}) \\
                         \end{array}
\right]
\end{eqnarray*}
where $\textbf{A}$ and $\textbf{B}$ are matrix representations of $A$ and $B$ in $M_{m}(K)$. The difference between this method and our method is in the choice of $\gamma_{{}_M}$ which for our case does not belong to $  F $. Though not explicitly mentioned in \cite{markin}, $\mathcal{M}_\mathcal{A} = \{A + \textbf{i}B ~\vert ~A,B \in \mathcal{A} \}$, with all the rules of multiplication as specified in the beginning of this section, is an {\it associative division algebra} since it contains no zero divisors. The centre $Z(\mathcal{M}_\mathcal{A})$ and maximal subfield $M$ of $\mathcal{M}_\mathcal{A}$ are as follows.
\begin{enumerate}
 \item When $\tau  = \sigma^{^{\frac{m}{2}}} \in \langle \sigma \rangle$, $Z(\mathcal{M}_\mathcal{A}) = \{a + \textbf{ij}^{\frac{m}{2}}b ~\vert~ a,b \in F\}$ and is isomorphic to $F(\sqrt{\gamma\gamma_{{}_M}})$, while $M = \{a + \textbf{ij}^{\frac{m}{2}}b~\vert~ a,b \in K\}$ and is isomorphic to $K(\sqrt{\gamma\gamma_{{}_M}})$. 
\item When $\tau \notin \langle \sigma \rangle$. $Z(\mathcal{M}_\mathcal{A}) = F \cap L$ and $M = K$ in which case, $\mathcal{M}_\mathcal{A}$ is a crossed product division algebra of degree $2m$ over $F \cap L$. 
\end{enumerate}
So, the map 
\begin{eqnarray*}
 \Phi_{\gamma_{{}_M}} :  \mathcal{M}_\mathcal{A} & \longrightarrow & M_{2m}(K) \\
  A+\textbf{i}B & \mapsto &\left[ \begin{array}{rr}
                          \textbf{A} & \gamma_{{}_M}\Upsilon(\textbf{B}) \\
                          \textbf{B} & \Upsilon(\textbf{A}) \\
                         \end{array}
\right]
\end{eqnarray*}
actually defines an injective homomorphism, and $\textrm{im}(\Phi_{\gamma_{{}_M}})$ is a matrix representation of the division algebra $\mathcal{M}_\mathcal{A}$ in $M_{2m}(K)$.

\subsection{The full-diversity STBC construction technique of \cite{unger}}\label{subsec_b}
In \cite{unger}, nonassociative quaternion division algebras are used to obtain rate-2 STBCs for 2 transmit antennas and rate-1 STBCs for 4 transmit antennas. The required nonassociative quaternion division algebra is constructed using a number field $L$ and its quadratic extension $K$ with $Gal(K/L) = \{1,\tau \}$. The algebra is given as $\{ a + \textbf{i}b ~ \vert ~a,b\in K, \textbf{i}^2 = \gamma_{{}_M} \in K\setminus L \}$, and is denoted by $Cay(K,\gamma_{{}_M})$. Since $\gamma_{{}_M}\notin L$, $Cay(K,\gamma_{{}_M})$ is nonassociative (When $\gamma_{{}_M} \in L$, the algebra is an associative quaternion algebra). Note that this is a special case of our proposed method, with $n=2$, $m=1$, $\mathcal{A} =K$, with the only difference being that $\gamma_{{}_M} \in K \setminus L$. Such a case is also included in our framework (see \eqref{d02}), but the resulting algebra is clearly also not power associative\footnote{An algebra $\mathcal{A}$ is said to be power associative if and only if $A^aA^b = A^{a+b}$, $\forall A \in \mathcal{A}$ and $\forall a,b \in  \{1,2,3,\cdots \}$.} since $\textbf{i}^2\textbf{i} = \gamma_{{}_M} \textbf{i} \neq \textbf{i}\gamma_{{}_M} =  \textbf{i}\left(\textbf{i}^2\right)$. A nonassociative algebra is division if and only if it has no zero divisors \cite{RDS}. It is easy to see that $Cay(K,\gamma_{{}_M})$ is a division algebra for if it were not, it would contain zero divisors so that $a \tau(a) = \gamma_{{}_M}$ for some $a \in K$ but $\gamma_{{}_M} \in K \setminus F$ and $a\tau(a) \in L$, contradicting the fact that $a \tau(a) = \gamma_{{}_M}$. The invertibility of the obtained matrices (nonzero) which are of the form in \eqref{M_form1} with $n=2,m=1$, follows from Theorem \ref{Thm2}.
 
\section{STBC construction}\label{sec_STC}
\subsection{General design procedure}

The general scheme to obtain invertible matrices as codewords of an STBC for $nm$ transmit antennas is as follows.
\begin{enumerate}
 \item $L$ is chosen to be either $\mathbb{Q}(i)$ or $\mathbb{Q}(\omega)$, the reason being that a finite subset of $\mathbb{Z}[i]$ is the QAM constellation and that of $\mathbb{Z}[\omega]$ is the HEX constellation, both of practical significance.
\item A cyclic division algebra $\mathcal{A} = (K/F,\sigma,\gamma)$ of degree $m$ over a number field $F$ with $F \neq L$, and an element $\gamma_{{}_M}$ are chosen such that 
\begin{enumerate}
\item $K/L$ is a Galois extension of degree $n$ with $Gal(K/L) = \langle \tau \rangle$.
\item $\sigma$ and $\tau$ commute.
\item $\gamma \in F \cap L $.
\item $\gamma_{{}_M} \in L \setminus F$.
\item $\prod_{i=0}^{n-1}\Upsilon^i(C) \neq \gamma_{{}_M}$, $\forall C \in \mathcal{A}$.
\end{enumerate}
\end{enumerate}
When $\mathcal{A}$ satisfies the above conditions, any nonzero matrix having the structure shown in \eqref{M_form1} is invertible, with $\textbf{A}_0$ and $\textbf{A}_1$ being matrix representations (in $M_{m}(K)$) of elements $A_0$ and $A_1$ of $\mathcal{A}$. If $\mathcal{A}$ is of degree $m$ over $F$ so that $[K:F] = m$, then $\textbf{A}_0, \textbf{ A}_1 \in K^{m\times m}$ and so, $\textbf{M} \in K^{nm\times nm}$. Each entry of $\textbf{A}_0$ and $\textbf{A}_1$ which belongs to $K$ can be viewed as a linear combination of $n$ independent elements over $L$ (since $[K:L] = n$). We express each element of $\textbf{A}_0$ and $\textbf{A}_1$ as a linear combination of some chosen $L$-basis $\{ \theta_i, i = 1,\cdots,n \vert \theta_i \in \mathcal{O}_{K}\}$ over $\mathcal{O}_{L}$. From the point of view of space-time coding, each codeword matrix of the STBC constructed using the proposed method has the structure shown in \eqref{M_form1} where $\textbf{A}_0$ and $\textbf{A}_1$ specifically have the structure given in \eqref{A_str} at the top of the page with $s_{ki}$, $k=0,1$, $i = 1,\cdots,nm$, being the complex information symbols taking values from QAM (a finite subset of $\mathbb{Z}[i]$) or HEX (a finite subset of $\mathbb{Z}[\omega]$) constellations. 
 
\begin{proposition}
 The rate of the STBC whose codeword matrices have the structure given in \eqref{M_form1} is $2$ complex symbols per channel use.
\end{proposition}
  \begin{proof}
   The STBC encodes $2nm$ independent complex symbols in $nm$ channel uses, hence allowing a rate of $2$ complex symbols per channel use. 
  \end{proof}

\begin{proposition}\label{prop_stbc_scheme}
 The STBC-scheme that is based on the STBCs constructed using the proposed method has the NVD property if $\gamma \in \mathcal{O}_F$ and $\gamma_{{}_M} \in \mathcal{O}_{L}$. 
\end{proposition}
\begin{proof}
 With $L = \mathbb{Q}(i)$ or $\mathbb{Q}(\omega)$, the ring of integers in $L$ is either $\mathbb{Z}[i]$ or $\mathbb{Z}[\omega]$. When $\gamma \in \mathcal{O}_F$ and $\gamma_{{}_M} \in \mathcal{O}_{L}$, application of Corollary \ref{cor1} establishes that for the infinite STBC whose codewords have the form shown in \eqref{M_form1} with $\textbf{A}_0$ and $\textbf{A}_1$ given by \eqref{A_str} and the symbols taking values from either $\mathbb{Z}[i]$ or $\mathbb{Z}[\omega]$ (see Definition \ref{nvd_def}), the determinant of any nonzero codeword matrix lies in $\mathbb{Z}[i]$ or $\mathbb{Z}[\omega]$. The result of the proposition follows.
\end{proof}

While our proposed scheme can be applied to a wide range of MIMO configurations, we illustrate its application to $4$ MIDO configurations\footnote{While the constructed STBCs can be used for arbitrary number of receive antennas, they are full-rate only for MIDO systems.} - $4\times2$, $6\times2$, $8\times2$ and $12 \times 2$ systems. The reason for choosing these $4$ configurations is easy to see - the existence of perfect codes \cite{ORBV} for $2,3,4,6$ transmit antennas and the Alamouti code for $2$ transmit antennas. The perfect codes of \cite{ORBV} are known for their large coding gain while the Alamouti code has the least ML-decoding complexity among STBCs from CDAs in addition to having the best coding gain among known rate-$1$ codes for the $2 \times 1$ system. We wish to combine the advantages of both these STBCs and so, we focus on the four mentioned MIDO systems. The STBC design procedure for these four MIMO configurations is briefly outlined as follows, and explicit code constructions are presented in the following subsections. For $n_t = 4, 8$, we choose $L$ to be $\mathbb{Q}(i)$ while for $n_t = 6,12$, we choose $L$ to be $\mathbb{Q}(\omega)$. $K$ and $\gamma_{{}_M}$ are respectively chosen to be the maximal subfield and the non-norm element of the division algebra used to construct the perfect codes for $n_t/2$ transmit antennas. So, $K$ is of the form $L(\theta)$, $ \theta \in \mathbb{R}$. Next, $\mathcal{A}$ is chosen to be $\mathcal{A} = (K/\mathbb{Q}(\theta),\sigma:i \mapsto -i,-1)$ which is a subalgebra of Hamilton's quaternion algebra $\mathcal{A}_{\mathbb{H}} = (\mathbb{C}/\mathbb{R},\sigma:i\mapsto -i,-1)$. The explicit code construction is illustrated in the following subsections.

\subsection{$4 \times 2$ MIDO system}\label{4x2}

We choose $L=\mathbb{Q}(i)$, $K=\mathbb{Q}(i,\sqrt{5})$ and $\gamma_{{}_M} = i$. The Galois group of $\mathbb{Q}(i,\sqrt{5})/\mathbb{Q}(i)$ is $\{1,\tau : \sqrt{5} \mapsto -\sqrt{5} \}$, and $(\mathbb{Q}(i,\sqrt{5})/\mathbb{Q}(i),\tau,i)$ is the CDA used to construct the Golden code for 2 transmit antennas. $\mathcal{A}$ is chosen to be $(\mathbb{Q}(i,\sqrt{5})/\mathbb{Q}(\sqrt{5}),\sigma,-1)$. Note that $\gamma_{{}_M} = i \notin \mathbb{Q}(\sqrt{5})$. The STBC for the $4\times2$ system (unnormalized with respect to SNR) obtained upon application of the construction scheme depicted in the previous section is given as
\begin{equation*}
 \mathcal{S}_{4\times2} = \left\{ \left[ \begin{array}{rrrr}
                                          a_0 & -\sigma(a_1) & i\tau(a_2) & -i\tau\sigma(a_3) \\
                                          a_1 & \sigma(a_0) & i\tau(a_3) & i\tau\sigma(a_2) \\
a_2 & -\sigma(a_3) & \tau(a_0) & -\tau\sigma(a_1) \\
a_3 & \sigma(a_2) & \tau(a_1) & \tau\sigma(a_0) \\

                                         \end{array} \right]
 \right\}
\end{equation*}
where $a_0 = s_{01}\theta_1 + s_{02}\theta_2$, $a_1 = s_{03}\theta_1 + s_{04}\theta_2$, $a_2 = s_{11}\theta_1 + s_{12}\theta_2$, $a_3 = s_{13}\theta_1 + s_{14}\theta_2$ with $s_{kj} \in M\textrm{-QAM} \subset \mathbb{Z}[i]$, and $\{\theta_1, \theta_2 \vert \theta_i \in \mathcal{O}_{K}\}$ is a suitable $\mathbb{Q}(i)$-basis of $\mathbb{Q}(i,\sqrt{5})$. From \cite{ORBV}, we pick $\theta_1 = \alpha, \theta_2 = \alpha \theta$ where $\alpha = 1 + i(1-\theta)$, $\theta = (1+\sqrt{5})/2$, and $\{ \alpha, \alpha\theta\}$ is now a basis of a principal ideal of $\mathcal{O}_{K}$ generated by $\alpha$. We now wish to prove that the STBC-scheme that is based on $\mathcal{S}_{4\times2}$ has the NVD property. To do so, it is sufficient from Proposition \ref{prop_stbc_scheme} to prove that $A\Upsilon(A) \neq i$, $\forall A \in \mathcal{A}$.
\begin{proposition}\label{lem4}
 Let $\mathcal{A} = (\mathbb{Q}(i,\sqrt{5})/\mathbb{Q}(\sqrt{5}),\sigma,-1)$. Then, $A\Upsilon(A) \neq i$, $\forall A \in \mathcal{A}$.
\end{proposition}

The proof of Proposition \ref{lem4} is given in Appendix \ref{app_4a}. So, $\mathcal{S}_{4\times2}$ is a rate-2 STBC with full-diversity and equipped with the property of non-vanishing determinant. 

\subsubsection{Minimum determinant}
When $s_{ki}$, $k=0,1$, $i = 1,\cdots, 4$, take values from $\mathbb{Z}[i]$, from Corollary \ref{cor1} the determinant of each of the codewords of $\mathcal{S}_{4\times2}$ belongs to $\mathbb{Z}[i]$. So, the minimum determinant of the unnormalized code is at least $1$. However, noting that the entries of the $i^{th}$ column of a codeword matrix, $i=1,\cdots,4$, are all respectively multiples of $\alpha$, $\sigma(\alpha)$, $\tau(\alpha)$, and $\sigma\tau(\alpha)$, the minimum determinant\footnote{The entries of the Golden code are not just from $\mathcal{O}_K = \mathbb{Z}[i,\theta]$ but from a principal ideal in $\mathcal{O}_K$ generated by $\alpha$. For details on this theory, one can refer to \cite{ORBV}.} is a multiple of $\vert \alpha\sigma(\alpha)\tau(\alpha)\sigma\tau(\alpha) \vert ^2 = \vert N_{K/L}(\alpha) \vert^4 = 25$ ($\sigma$ is simply complex conjugation). When $s_{ki}$ take values from an $M$-QAM with average energy $E$ units, a normalization factor (see Note \ref{note1} in Section \ref{sec_system_model}) of $\frac{1}{\sqrt{4E\vert \alpha \vert ^2 (1+\theta^2)}} = \frac{1}{\sqrt{20E}}$ has to be taken into account. Further, since the difference between any two signal points in a QAM constellation is a multiple of $2$, the normalized minimum determinant of $\mathcal{S}_{4\times2}$ is $\delta_{min}(\mathcal{S}_{4\times2})=25\left(\frac{2}{\sqrt{20E}}\right)^8 = \frac{1}{25E^4}$. 

\begin{note}
 The STBC for the $4 \times 2$ MIDO system in \cite[Section IV-B]{markin} makes use of $\mathcal{A} = ( \mathbb{Q}(i, \sqrt{5})/ \mathbb{Q}(i),\sigma:\sqrt{5} \mapsto -\sqrt{5}, i)$, $\tau = \sigma$, and $\gamma_{{}_M} = 1- i$. The constructed STBC has NVD, but lower normalized minimum determinant than $\mathcal{S}_{4\times2}$. Its normalized minimum determinant can be calculated to be $\frac{256}{25^3E^4}$, which is $\left(\frac{4}{5}\right)^4 = 0.4096$ times the normalized minimum determinant of $\mathcal{S}_{4 \times 2}$, which corresponds to a coding gain that is $0.8$ times the coding gain of $\mathcal{S}_{4 \times 2}$. 
\end{note}
  
\subsubsection{Relation with Srinath-Rajan code}
A codeword matrix of the SR-code is given by 
\begin{equation*}
 \textbf{S} = \left[\begin{array}{rr}
                     \textbf{A} & e^{\frac{i\pi}{4}}\textbf{C} \\
                     e^{\frac{i\pi}{4}}\textbf{B} & \textbf{D} \\
                    \end{array}\right]
\end{equation*}
where 

\begin{eqnarray*}
 \textbf{A} & = & \left[\begin{array}{cc}
                     x_{1I}+ix_{3Q} & -x_{2I}+ix_{4Q} \\
                     x_{2I}+ix_{4Q} & x_{1I}-ix_{3Q} \\
                    \end{array}\right],\\
\textbf{B} & = & \left[\begin{array}{cc}
                     x_{5I}+ix_{7Q} & -x_{6I}+ix_{8Q} \\
                     x_{6I}+ix_{8Q} & x_{5I}-ix_{7Q} \\
                    \end{array}\right],
\end{eqnarray*}
\begin{eqnarray*}
\textbf{C} & = & \left[\begin{array}{cc}
                     x_{7I}+ix_{5Q} & -x_{8I}+ix_{6Q} \\
                     x_{8I}+ix_{6Q} & x_{7I}-ix_{5Q} \\
                    \end{array}\right],\\
\textbf{D} & = & \left[\begin{array}{cc}
                     x_{3I}+ix_{1Q} & -x_{4I}+ix_{2Q} \\
                     x_{4I}+ix_{2Q} & x_{3I}-ix_{1Q} \\
                    \end{array}\right]
\end{eqnarray*}
with $x_{iI}$ and $x_{iQ}$ being the real and imaginary parts respectively of the complex symbol $x_i$, $i=1,\cdots,8$, and $x_i = e^{\frac{i\tan^{-1}(2)}{2}}s_i$, $s_i \in \mathbb{Z}[i]$ (i.e., from a suitable QAM constellation). Denoting $\tan^{-1}(2)/2$ by $\theta_g$, we have
\begin{eqnarray*}
x_i & = & \cos \theta_gs_{iI} -\sin \theta_gs_{iQ} + i( \sin \theta_gs_{iI} +\cos \theta_gs_{iQ}) \\ 
& = & \sin \theta_g \left[ \cot \theta_gs_{iI} -s_{iQ} + i(s_{iI} +\cot \theta_gs_{iQ}) \right] \\
& = & \sin \theta_g \left[ \theta s_{iI} -s_{iQ} + i(s_{iI} +\theta s_{iQ}) \right] 
\end{eqnarray*}
where $\theta = (1+\sqrt{5})/2$. So, it is easy to work out that $\textbf{S} = \textbf{US}^\prime\textbf{U}^H\textbf{D}$ with $\textbf{U} = \textrm{diag}[1,1,e^{\frac{i\pi}{4}},e^{\frac{i\pi}{4}}]$, 

\begin{equation}\label{form_s}
\textbf{S}^\prime = \left[\begin{array}{rrrr}
                     f_0 & -\sigma(f_1) & i\tau(f_2) & -i\sigma\tau(f_3)  \\
                     f_1 &  \sigma(f_0) & i\tau(f_3) &  i\sigma\tau(f_2)\\
                     f_2 & -\sigma(f_3) & \tau(f_0) & -\sigma\tau(f_1) \\
                     f_3 &  \sigma(f_2) & \tau(f_1) &  \sigma\tau(f_0)
                    \end{array}\right] 
\end{equation}
where $f_0 = -s_{1Q}+is_{3I} + \theta(s_{1I}+is_{3Q})$, $f_1 = -s_{2Q}+is_{4I} + \theta(s_{2I}+is_{4Q})$, $f_2 = -s_{5Q}+is_{7I} + \theta(s_{5I}+is_{7Q})$, $f_3 = -s_{6Q}+is_{8I} + \theta(s_{6I}+is_{8Q})$, and $\textbf{D} = \textrm{diag}[\sin \theta_g,\sin \theta_g, \theta \sin \theta_g,\theta \sin \theta_g  ]$. So, $f_i \in \mathbb{Z}[i,\theta]$ and it is easy to observe that a codeword matrix $\textbf{S}$ of $\mathcal{S}_{4\times2}$ constructed in this subsection has the structure $\textbf{S} = \textbf{S}^{\prime \prime} \textbf{D}_1$ where $\textbf{S}^{\prime \prime}$ has the same algebraic structure as $\textbf{S}^\prime$ in \eqref{form_s} and $\textbf{D} = \textrm{diag}[\frac{\alpha}{\sqrt{5}},\frac{\sigma(\alpha)}{\sqrt{5}},\frac{\tau(\alpha)}{\sqrt{5}},\frac{\sigma\tau(\alpha)}{\sqrt{5}} ]$ (the scaling factor of $1/\sqrt{5}$ is for energy equalization). Clearly, the SR-code and $\mathcal{S}_{4\times2}$ have the same underlying algebraic structure and hence the same minimum determinant (this follows from the fact that $\vert det(\textbf{D}) \vert = \vert det(\textbf{D}_1) \vert = 1/5 $) and ML-decoding complexity. This also establishes that the STBC-scheme that is based on the SR-code has the NVD property, which had been previously only conjectured. 

\subsection{$6 \times 2$ MIDO system} 

For this MIDO configuration, we choose $L=\mathbb{Q}(\omega)$, $K=\mathbb{Q}(\omega,\theta)$ and $\gamma_{{}_M} = \omega$, where $\theta = \zeta_7 + \zeta_7^{-1} = 2\cos\left(\frac{2\pi}{7}\right)$ with $\zeta_7$ denoting the primitive $7^{th}$ root of unity. Note that $(\mathbb{Q}(\omega,\theta)/\mathbb{Q}(\omega),\tau, \omega)$ is the CDA used to construct the perfect code for 3 transmit antennas \cite{ORBV} with $\tau$ given by $\tau : \zeta_7 + \zeta_7^{-1} \mapsto \zeta_7^2 + \zeta_7^{-2}$. $\mathcal{A}$ is chosen to be $(\mathbb{Q}(\omega,\theta)/\mathbb{Q}(\theta),\sigma: i \mapsto -i, -1)$. Since $\omega = (-1+\sqrt{3}i)/2$, $\sigma(\omega) = \omega^2$. It is to be noted that $\gamma_{{}_M} = \omega \notin \mathbb{Q}(\theta)$. The rate-2 STBC (unnormalized with respect to SNR) for $6$ transmit antennas is given by
\begin{equation*}
 \mathcal{S}_{6\times2} = \left\{ \left[ \begin{array}{ccc}
                                          \textbf{A}_0 & \textbf{O} &  \omega \Upsilon^2(\textbf{A}_1)\\
                                          \textbf{A}_1 & \Upsilon(\textbf{A}_0) & \textbf{O} \\
                                          \textbf{O} & \Upsilon(\textbf{A}_1) & \Upsilon^2(\textbf{A}_0) \\
                                         \end{array} \right]
 \right\}
\end{equation*}
where
\begin{equation*}
 {\bf A}_k =  \left[ \begin{array}{cc}
                                         \sum_{i=1}^{3}s_{ki}\theta_i  & -\sigma\left(\sum_{i=1}^{3}s_{k(i+3)}\theta_i\right) \\
                                         \sum_{i=1}^{3}s_{k(i+3)}\theta_i & \sigma\left(\sum_{i=1}^{3}s_{ki}\theta_i\right) \\
                                         \end{array} \right] 
\end{equation*}
with $s_{kj} \in M\textrm{-HEX} \subset \mathbb{Z}[\omega]$, $k=0,1$, $i =1,\cdots,6$. Here, $\{ \theta_1, \theta_2, \theta_3\}$ is a basis of a principal ideal in $\mathcal{O}_{K}$ generated by $\theta_1$ with \cite{ORBV} $\theta_1 = 1+\omega+\theta$, $\theta_2 = -1-2\omega+\omega\theta^2$ and $\theta_3 = (-1-2\omega)+(1+\omega)\theta + (1+\omega)\theta^2$. To prove that the STBC-scheme which is based on $\mathcal{S}_{6\times2}$ has the NVD property, it is sufficient to prove that $A\Upsilon(A)\Upsilon^2(A) \neq \omega$, $\forall A \in \mathcal{A}$. 

\begin{proposition}\label{prop_a}
$A\Upsilon(A)\Upsilon^2(A) \neq \omega$, $\forall A \in (\mathbb{Q}(\omega,\theta)/\mathbb{Q}(\theta),$ $\sigma,-1)$. 
\end{proposition}

The proof of Proposition \ref{prop_a} is on similar lines to that of Proposition \ref{lem4} and given in Appendix \ref{app_4}.

\subsubsection{Minimum determinant}
When $s_{ki}$, $k=0,1$, $i = 1,\cdots, 6$, take values from $\mathbb{Z}[\omega]$, from Corollary \ref{cor1} the determinant of each of the codewords of $\mathcal{S}_{6\times2}$ belongs to $\mathbb{Z}[\omega]$. So, the minimum determinant of the unnormalized code is at least $1$. However, the perfect code for 3 antennas has its entries from a principal ideal in $\mathcal{O}_K$ generated by $\theta_1$. So, the minimum determinant is $\vert N_{K/L}(\theta_1)\vert^4 = 7^2 = 49$. When the constellation used is $M$-HEX (so that the difference between any two signal points is a multiple of 2), after taking into account a normalization factor of $1/\sqrt{4E(\vert \theta_1\vert^2 +\vert \theta_2\vert^2 +\vert \theta_3\vert^2) } = 1/\sqrt{28E}$, the normalized minimum determinant of $\mathcal{S}_{6\times2}$ is $49\left(\frac{2}{\sqrt{28E}}\right)^{12} = \frac{1}{7^4E^6}$. 

%%%%%%%%%%%%%%%%%%%%%%%%%%%%%%%%%%%%%%%%%%%%%%%%%%%%%%%%%%%%%%%%%%%%%%%%%%%%%%%%%%%%%%%%%%%%%%%%%%
\begin{table*}
\begin{equation}\label{12x2}
 \mathcal{S}_{12\times2} = \left\{ \left[ \begin{array}{cccccc}
                                          \textbf{A}_0 & \textbf{O} & \textbf{O}  & \textbf{O} & \textbf{O} & -\omega \Upsilon^5(\textbf{A}_1)\\
                                          \textbf{A}_1 & \Upsilon(\textbf{A}_0) & \textbf{O} &  \textbf{O} & \textbf{O} & \textbf{O}\\
                                          \textbf{O} & \Upsilon(\textbf{A}_1) & \Upsilon^2(\textbf{A}_0) & \textbf{O} & \textbf{O} & \textbf{O}\\
\textbf{O} & \textbf{O} & \Upsilon^2(\textbf{A}_1)&  \Upsilon^3(\textbf{A}_0) & \textbf{O} & \textbf{O}\\
                                          \textbf{O} & \textbf{O} & \textbf{O} &  \Upsilon^3(\textbf{A}_1) & \Upsilon^4(\textbf{A}_0) & \textbf{O}\\
                                          \textbf{O} & \textbf{O} & \textbf{O} &  \textbf{O} & \Upsilon^4(\textbf{A}_1) & \Upsilon^5(\textbf{A}_0)\\
                                         \end{array} \right], {\bf A}_k =  \left[ \begin{array}{cc}
                                         z_{k0}  & -\sigma\left(z_{k1}\right) \\
                                         z_{k1} & \sigma\left(z_{k0}\right) \\
                                         \end{array} \right] 
 \right\}.
\end{equation}
\end{table*}
%%%%%%%%%%%%%%%%%%%%%%%%%%%%%%%%%%%%%%%%%%%%%%%%%%%%%%%%%%%%%%%%%%%%%%%%%%%%%%%%%%%%%%%%%%

\subsection{$8 \times 2$ MIDO system}

$L$ is chosen to be $\mathbb{Q}(i)$, and $K=\mathbb{Q}(i,\theta)$, $\gamma_{{}_M} = i$, where $\theta = \zeta_{15} + \zeta_{15}^{-1} = 2\cos\left(\frac{2\pi}{15}\right)$ with $\zeta_{15}$ denoting the primitive $15^{th}$ root of unity. Note that with $\tau$ given as $\tau : \zeta_{15} + \zeta_{15}^{-1} \mapsto \zeta_{15}^2 + \zeta_{15}^{-2}$, $\left(\mathbb{Q}(i,\theta)/\mathbb{Q}(i),\tau,i \right)$ is the CDA used to construct the perfect code for 4 transmit antennas \cite{ORBV}. Next, $\mathcal{A}$ is chosen to be $(\mathbb{Q}(i,\theta)/\mathbb{Q}(\theta),\sigma,-1)$. It is to be noted that $\gamma_{{}_M} = i \notin \mathbb{Q}(\theta)$. The rate-2 STBC (unnormalized with respect to SNR) for $8$ transmit antennas is given by
\begin{equation*}
 \mathcal{S}_{8\times2} = \left\{ \left[ \begin{array}{cccc}
                                          \textbf{A}_0 & \textbf{O} & \textbf{O} &  i\Upsilon^3(\textbf{A}_1)\\
                                          \textbf{A}_1 & \Upsilon(\textbf{A}_0) & \textbf{O}  & \textbf{O} \\
                                          \textbf{O} & \Upsilon(\textbf{A}_1) & \Upsilon^2(\textbf{A}_0) & \textbf{O} \\
\textbf{O} & \textbf{O} &  \Upsilon^2(\textbf{A}_1) & \Upsilon^3(\textbf{A}_0) \\
                                         \end{array} \right]
 \right\}
\end{equation*}
where
\begin{equation*}
 {\bf A}_k =  \left[ \begin{array}{cc}
                                         \sum_{i=1}^{4}s_{ki}\theta_i  & -\sigma\left(\sum_{i=1}^{4}s_{k(i+4)}\theta_i\right) \\
                                         \sum_{i=1}^{4}s_{k(i+4)}\theta_i & \sigma\left(\sum_{i=1}^{4}s_{ki}\theta_i\right) \\
                                         \end{array} \right] 
\end{equation*}
with $s_{kj} \in M\textrm{-QAM} \subset \mathbb{Z}[i]$, $k=0,1$, $j =1,\cdots,8$. Here, $\{ \theta_1, \theta_2, \theta_3, \theta_4\}$ is a basis \cite{ORBV} of a principal ideal in $\mathcal{O}_{K}$ generated by $\theta_1$, where $\theta_1   = \alpha$, $\theta_2   =\alpha\theta$, $\theta_3 = \alpha\theta(-3+\theta^2)$, $\theta_4  = \alpha(-1-3\theta+\theta^2+\theta^3)$ with $\alpha = 1-3i+i\theta^2$. To prove that the STBC-scheme which is based on $\mathcal{S}_{8\times2}$ has the NVD property, it is sufficient from Proposition \ref{prop_stbc_scheme} to prove that $A\Upsilon(A)\Upsilon^2(A)\Upsilon^3(A) \neq i$, $\forall A \in \mathcal{A}$. 
\begin{proposition}\label{prop_b}
 Let $\mathcal{A} = (\mathbb{Q}(i,\theta)/\mathbb{Q}(\theta),\sigma,-1)$. Then, $A\Upsilon(A)\Upsilon^2(A)\Upsilon^3(A) \neq i$, $\forall A \in \mathcal{A}$.
\end{proposition}

The proof of Proposition \ref{prop_b} is given in Appendix \ref{app_5}.

\subsubsection{Minimum determinant}
When $s_{ki}$, $k=0,1$, $i = 1,\cdots, 8$, take values from $\mathbb{Z}[i]$, from Corollary \ref{cor1} the determinant of each of the codewords of $\mathcal{S}_{8\times2}$ belongs to $\mathbb{Z}[i]$ and hence the minimum determinant of the unnormalized code is at least $1$. However, the perfect code for 4 antennas has its entries from a principal ideal in $\mathcal{O}_K$ generated by $\theta_1$ whose norm $N_{K/L}(\theta_1)$ has modulus equal to $\sqrt{45}$. Hence, the minimum determinant is $  \vert N_{K/L}(\theta_1)\vert^4 = 45^2$. When the constellation used is $M$-QAM, after taking into account a normalization factor of $1/\sqrt{4E \sum_{i=1}^{4}\vert\theta_i \vert ^2 } = 1/\sqrt{60E}$, the normalized minimum determinant of $\mathcal{S}_{8\times2}$ is $(45^2)\left(\frac{2}{\sqrt{60E}}\right)^{16} = \frac{1}{25(15)^4E^8}$.

\subsection{$12 \times 2$ MIDO system}

We choose $L$ is chosen to be $\mathbb{Q}(\omega)$, and $K=\mathbb{Q}(\omega,\theta)$, $\gamma_{{}_M} = -\omega$, where $\theta = \zeta_{28} + \zeta_{28}^{-1} = 2\cos\left(\frac{\pi}{14}\right)$ with $\zeta_{28}$ denoting the primitive $28^{th}$ root of unity. With $\tau$ given as $\tau : \zeta_{28} + \zeta_{28}^{-1} \mapsto \zeta_{28}^2 + \zeta_{28}^{-2}$, $\left(\mathbb{Q}(\omega,\theta)/\mathbb{Q}(\omega),\tau, -\omega \right)$ is the CDA used to construct the perfect code for 6 transmit antennas \cite{ORBV}. $\mathcal{A}$ is chosen to be $(\mathbb{Q}(\omega,\theta)/\mathbb{Q}(\theta),\sigma,-1)$. It is clear that $\gamma_{{}_M} = -\omega \notin \mathbb{Q}(\theta)$. The rate-2 STBC (unnormalized with respect to SNR) for $12$ transmit antennas is given by \eqref{12x2} at the top of the next page with
\begin{equation*}
 \left[ \begin{array}{c}
         z_{ki}\\
         \tau(z_{ki})\\
         \tau^2(z_{ki})\\
         \tau^3(z_{ki})\\
         \tau^4(z_{ki})\\
         \tau^5(z_{ki})\\
        \end{array}
\right] =  \textbf{R}\left[ \begin{array}{c}
         s_{k(6i+1)}\\
         s_{k(6i+2)}\\
         s_{k(6i+3)}\\
         s_{k(6i+4)}\\
         s_{k(6i+5)}\\
         s_{k(6i+6)}\\
        \end{array}
\right], ~ k=0,1, ~ i=0,1,
\end{equation*}
where $s_{kj} \in M\textrm{-HEX} \subset \mathbb{Z}[\omega]$, $k=0,1$, $j =1,\cdots,24$, and $\textbf{R}$, defined by \eqref{integral_basis}, is obtained from \cite{ORBV} and shown in \eqref{R_12x2} at the top of the next page.

As done for the previous STBCs, to prove that the STBC-scheme that is based on $\mathcal{S}_{12\times2}$ has the NVD property, it is sufficient to show that $A\Upsilon(A)\Upsilon^2(A)\cdots \Upsilon^5(A) \neq -\omega$, $\forall A \in \mathcal{A}$.
\begin{proposition}\label{prop_c}
$A\Upsilon(A)\Upsilon^2(A)\cdots\Upsilon^5(A) \neq -\omega$, $\forall A \in$ $ (\mathbb{Q}(\omega,\theta)/\mathbb{Q}(\theta),\sigma,-1)$.
\end{proposition}

The proof of Proposition \ref{prop_c} is provided in Appendix \ref{app_6}.

%%%%%%%%%%%%%%%%%%%%%%%%%%%%%%%%%%%%%%%%%%%%%%%%%%%%%%%%%%%%%%%%%
\begin{figure*}
\begin{equation}\label{R_12x2}
\textbf{R} = \left[\begin{array}{rrrrrr}
 1.9498         &  1.3019-0.8660i & -0.0549-0.8660i &
-1.7469-0.8660i &  1.5636         &  0.8677  \\
 0.8677         & -1.7469-0.8660i &  1.3019-0.8660i &
-0.0549-0.8660i & -1.9498         &  1.5636 \\
 1.5636         & -0.0549-0.8660i & -1.7469-0.8660i &
 1.3019-0.8660i & -0.8677         & -1.9498 \\
-1.9498         &  1.3019-0.8660i & -0.0549-0.8660i &
-1.7469-0.8660i & -1.5636         & -0.8677 \\
-0.8677         & -1.7469-0.8660i &  1.3019-0.8660i &
-0.0549-0.8660i &  1.9498         & -1.5636 \\
-1.5636         & -0.0549-0.8660i & -1.7469-0.8660i &
 1.3019-0.8660i &  0.8677         &  1.9498\\
\end{array}\right].
\end{equation}
\hrule
\end{figure*}
%%%%%%%%%%%%%%%%%%%%%%%%%%%%%%%%%%%%%%%%%%%%%%%%%%%%%%%%%%%%%%%%%

\subsubsection{Minimum determinant}
From Corollary \ref{cor1}, the minimum determinant of the unnormalized code is $1$. Since the entries of the perfect code for 6 antennas are not in a principal ideal, a lower bound on the minimum determinant of the unnormalized code is 1. It can be checked that the norm of each row of $\textbf{R}$ is $\sqrt{14}$. So, taking into account a normalization factor of $1/\sqrt{(4)(14)E} = \sqrt{56E}$, the normalized minimum determinant of $\mathcal{S}_{12\times2}$ whose symbols take values from $M$-HEX is at least $\left(\frac{1}{\sqrt{14E}}\right)^{24} = \left(\frac{1}{14E}\right)^{12}$.

\begin{note}
Among the four STBCs constructed, $\mathcal{S}_{4\times2}$ and $\mathcal{S}_{8\times2}$ have orthogonal generator matrices (see Definition \ref{gen_mat} and Definition \ref{cubic_shaping}) and hence have cubic shaping while the other two codes do not. 
\end{note}

\section{ML-decoding Complexity}\label{sec_ml_comp}
In this section, we analyze the ML-decoding complexity of the constructed STBCs as a function of the constellation size $M$ which is assumed to be a square integer. Consider the ML-decoding metric given by $\Vert \textbf{Y}-\sqrt{\rho} \textbf{HS} \Vert^2$ which is to be minimized over all possible codewords $\textbf{S} \in \mathcal{S}$. We have 
\begin{eqnarray*}
 \Vert \textbf{Y}-\sqrt{\rho} \textbf{HS} \Vert^2& = &tr\left[ (\textbf{Y}-\sqrt{\rho}\textbf{HS})(\textbf{Y}-\sqrt{\rho} \textbf{HS})^H\right]\\
&  = & tr\left(\textbf{YY}^H - \sqrt{\rho}\textbf{YS}^H\textbf{H}^H \right.\\
&&  -\left. \sqrt{\rho}\textbf{HSY}^H + \rho\textbf{HSS}^H\textbf{H}^H \right)\\
& = &tr\left(\textbf{YY}^H\right)- 2\sqrt{\rho}\mathcal{R}e\left\{tr\left(\textbf{HSY}^H\right)\right\}\\
&&+ \rho~ tr\left(\textbf{HSS}^H\textbf{H}^H \right).
\end{eqnarray*}
Now, expressing a codeword $\textbf{S}$ as $\sum_ {i=1}^k \left(\bar{s}_{i}\bar{\textbf{A}}_{i} + \check{s}_{i}\check{\textbf{A}}_{i}\right)$ (see \eqref{form_stbc}), where $\bar{s}_{i},\check{s}_{i} \in \sqrt{M}$-PAM, we get 
\begin{eqnarray*}
tr\left(\textbf{HSY}^H\right) & = & \sum_{i=1}^k \left(\bar{s}_{i} tr\left({\bf H}\bar{\textbf{A}}_{i}{\bf Y}^H\right) + \check{s}_{i} tr\left({\bf H}\check{\textbf{A}}_{i}{\bf Y}^H\right) \right),
\end{eqnarray*}
\begin{eqnarray*}
 tr\left(\textbf{HSS}^H\textbf{H}^H \right) & = & \sum_{i=1}^{k}n_{ii}\bar{s}_{i}^2  + \sum_{i=1}^{k}m_{ii}\check{s}_{i}^2 \\ && + \sum_{i=1}^{k-1}\sum_{j=i+1}^{k}p_{ij}\bar{s}_{i}\bar{s}_{j}
+ \sum_{i=1}^{k-1}\sum_{j=i+1}^{k}q_{ij}\check{s}_{i}\check{s}_{j} \\
%\end{eqnarray*}
%\begin{eqnarray*}
 && + \sum_{i=1}^{k}\sum_{j=1}^{k}r_{ij}\bar{s}_{i}\check{s}_{j} 
\end{eqnarray*}
where 
\begin{eqnarray*}
  && n_{ii} =   tr\left({\bf H}\bar{\textbf{A}}_{i}\bar{\textbf{A}}_{i}^H{\bf H}^H\right),~m_{ii} = tr\left({\bf H}\check{\textbf{A}}_{i}\check{\textbf{A}}_{i}^H{\bf H}^H\right),\\
 &&p_{ij} = tr\left({\bf H}\left(\bar{\textbf{A}}_{i}\bar{\textbf{A}}_{j}^H +\bar{\textbf{A}}_{j}\bar{\textbf{A}}_{i}^H\right){\bf H}^H\right),\\
&&q_{ij}=  tr\left({\bf H}\left(\check{\textbf{A}}_{i}\check{\textbf{A}}_{j}^H+\check{\textbf{A}}_{j}\check{\textbf{A}}_{i}^H\right){\bf H}^H\right),\\
&&r_{ij} = tr\left({\bf H}\left(\bar{\textbf{A}}_{i}\check{\textbf{A}}_{j}^H+\check{\textbf{A}}_{j}\bar{\textbf{A}}_{i}^H\right){\bf H}^H\right). 
\end{eqnarray*}

\noindent Therefore, the only term in the ML-decoding metric that has an entanglement of the information symbols is $\rho ~tr\left(\textbf{HSS}^H\textbf{H}^H \right)$. Hence, $\textbf{SS}^H$ defines the ML-decoding complexity of the STBC. Now, let $z_1 = \sum_{i=1}^{n}s_{1i}\theta_i$, $z_2 = \sum_{i=1}^{n}s_{2i}\theta_i$, where $s_{ki}$ take values from either QAM or HEX constellations, $\theta_i \in \mathbb{C}$. If the transmitted codeword is
\begin{equation*}
\textbf{S} = \left[ \begin{array}{rr}
                     z_1  & -\sigma(z_2) \\
                     z_2  &  \sigma(z_1) \\
                    \end{array}
\right]
\end{equation*}
where $\sigma$ performs complex conjugation, $\textbf{SS}^H = (\vert z_1 \vert^2 + \vert z_2 \vert^2)\textbf{I}$. Hence, the group of symbols $\{s_{1i},i=1,\cdots,n \}$ that $z_1$ consists of are disentangled in the decoding metric from $\{s_{2i},i=1,\cdots,n \}$ that $z_2$ consists of. So, $\{s_{1i},i=1,\cdots,n \}$ can be decoded independently of $\{s_{2i},i=1,\cdots,n \}$. In addition, if $s_{ki}$ take values from a square QAM constellation and $\theta_i$, $i=1,\cdots,n$, are of the form $\theta_i = \alpha \theta_i^\prime$ where $\alpha \in \mathbb{C}$ and $\theta_i^\prime \in \mathbb{R}$, then within each group $\{s_{ki},i=1,\cdots,n \}$, $k=1,2$, the group comprising the real parts of each symbol is separable from the group comprising the imaginary parts. Hence, 
\begin{enumerate}
 \item when $s_{ki}$ take values from a square QAM, the four groups - $\{\mathcal{R}e(s_{1i}),i=1,\cdots,n \}$, $\{\mathcal{I}m(s_{1i}),i=1,\cdots,n \}$, $\{\mathcal{R}e(s_{2i}),i=1,\cdots,n \}$ and $\{\mathcal{I}m(s_{2i}),i=1,\cdots,n \}$, are independently decodable of one another.
\item when $s_{ki}$ take values from a HEX constellation, the two groups  $\{s_{1i},i=1,\cdots,n \}$ and $\{s_{2i},i=1,\cdots,n \}$ are independently decodable of one another.
\end{enumerate}
So, we have the following proposition.
\begin{proposition}\label{ml_prop}
 Let the codeword matrices of an STBC $\mathcal{S}$ be block diagonal of the form $ \textbf{S} = \textrm{diag}\left[\textbf{A}, \Upsilon(\textbf{A}), \cdots, \Upsilon^{n-1}(\textbf{A}) \right]$ where 
\begin{eqnarray*}
\textbf{A} & = & \left[ \begin{array}{rr}
                     \sum_{i=1}^{n}s_{1i}\theta_i  & -\sigma(\sum_{i=1}^{n}s_{2i}\theta_i) \\
                     \sum_{i=1}^{n}s_{2i}\theta_i &  \sigma(\sum_{i=1}^{n}s_{1i}\theta_i) \\
                    \end{array}
\right], \end{eqnarray*}

\begin{eqnarray*}
\Upsilon^l(\textbf{A}) & = &\left[ \begin{array}{rr}
                     \tau^{l}\left(\sum_{i=1}^{n}s_{1i}\theta_i\right)  & -\sigma\tau^{l}\left(\sum_{i=1}^{n}s_{2i}\theta_i \right) \\
                     \tau^{l}\left(\sum_{i=1}^{n}s_{2i}\theta_i\right) &  \sigma\tau^{l}\left(\sum_{i=1}^{n}s_{1i}\theta_i\right)) \\
                    \end{array}
\right]
\end{eqnarray*}
and $\{\theta_i, i=1,\cdots,n \vert \theta_i \in \mathcal{O}_K\}$ is a $\mathbb{Q}(i)$-basis (or $\mathbb{Q}(\omega)$-basis) of a number field $K$ which is a Galois extension of degree $n$ over $\mathbb{Q}(i)$ (respectively $\mathbb{Q}(\omega)$) with Galois group $\langle \tau \rangle$. If $\theta_i$, $i=1,\cdots,n$, are of the form $\theta_i = \alpha \theta_i^\prime$ where $\alpha \in \mathbb{C}$ and $\theta_i^\prime \in \mathbb{R}$, then $\mathcal{S}$ is 
\begin{enumerate}
 \item four-group decodable if $s_{ki}$ take values from QAM.
\item two-group decodable if $s_{ki}$ take values from HEX. 
\end{enumerate}
\begin{proof}
 The proof is trivial and follows from the argument preceding the proposition.
\end{proof}
\end{proposition}

Following Proposition \ref{ml_prop}, the ML-decoding complexity of the codes constructed in this paper is easy to analyze. We express a codeword matrix of the STBC as $\textbf{S} = \textbf{S}(\textbf{A}_0) + \textbf{S}(\textbf{A}_1)$, where $\textbf{S}(\textbf{A}_0) = \textrm{diag}\left[\textbf{A}_0, \Upsilon(\textbf{A}_0), \cdots, \Upsilon^{n-1}(\textbf{A}_0) \right]$, $n = n_t/2$, and 
\begin{eqnarray*}
 \textbf{S}(\textbf{A}_1) = \left[ \begin{array}{cccc}
                                    \textbf{O} & \textbf{O} & \ldots &  \gamma_{{}_M}\Upsilon^{n-1}(\textbf{A}_1)\\
                                          \textbf{A}_1 & \textbf{O} & \ldots  & \textbf{O} \\
                                          \textbf{O} & \Upsilon(\textbf{A}_1) & \ldots & \textbf{O} \\
 \vdots & \vdots & \ddots & \vdots \\
\textbf{O} & \textbf{O} &  \Upsilon^{n-2}(\textbf{A}_1) & \textbf{O}\\ 
                                   \end{array}
\right].
\end{eqnarray*}
Both $\textbf{S}(\textbf{A}_0)$ and $\textbf{S}(\textbf{A}_0)$ contain $n_t$ complex information symbols each. So,
{\small
 \begin{eqnarray*}
 \min_{\textbf{S} \in \mathcal{S}} \left\{ \Vert \textbf{Y}-\sqrt{\rho} \textbf{HS} \Vert^2\right\} & = & \min_{\textbf{S}(\textbf{A}_1)} \left\{ \min_{\textbf{S}(\textbf{A}_0)}\left \{ \Vert \textbf{Y}^\prime - \sqrt{\rho} \textbf{HS}(\textbf{A}_0) \Vert^2 \right \}  \right \} \\          \end{eqnarray*}
}

\noindent where $ \textbf{Y}^\prime  = \textbf{Y} - \rho\textbf{HS}(\textbf{A}_1)$. Assuming that the complex constellations used are $M$-QAM or $M$-HEX, from Proposition \ref{ml_prop}, it can be noted that calculating $\min_{\textbf{S}(\textbf{A}_0)}\left \{ \Vert \textbf{Y}^\prime - \sqrt{\rho} \textbf{HS}(\textbf{A}_0) \Vert^2 \right \}$ requires only $\mathcal{O}\left(M^{\frac{n_t}{4}}\right)$ calculations (for $n_t =4,8$) and $\mathcal{O}\left(M^{\frac{n_t}{2}}\right)$ calculations (for $n_t =6,12$). Therefore, the overall ML-decoding complexity of the STBCs is $\mathcal{O}\left(M^{n_t} ~ M^{\frac{n_t}{4}}\right) =  \mathcal{O}\left(M^{\frac{5n_t}{4}}\right)$ for $n_t = 4,8$, while for $n_t = 6,12$, it is $\mathcal{O}\left(M^{n_t} ~ M^{\frac{n_t}{2}}\right) =  \mathcal{O}\left(M^{\frac{3n_t}{2}}\right)$. In addition, hard-limiting (see \cite{pav_rajan2} for details) further reduces the overall ML-decoding complexity by a factor of $\sqrt{M}$. Table \ref{tab1} captures the salient features of the constructed codes along with their comparison with some of the best known STBCs.

%So, the ML-decoding complexity of $\mathcal{S}_{4\times2}$ is of the order of $(M^4)(M)$ - the factor $M^4$ is due to the joint evaluation of the last $4$ complex symbols and the factor $M$ is due to independence in evaluation of each of the first $4$ symbols once the last $4$ are fixed. In addition, by employing hard-limiting \cite{pav_rajan2} to evaluate the real part of each of the first $4$ symbols, the complexity can be further reduced by a factor of $\sqrt{M}$ ($M$ is assumed to be a square). Hence, the ML-decoding complexity of $\mathcal{S}_{4\times2}$ is of the order of $M^{4.5}$. By a similar analysis, the ML-decoding complexity of $\mathcal{S}_{8\times2}$ is of the order of $(M^8)(M^{1.5}) = M^{9.5}$.
  
%For $\mathcal{S}_{6\times2}$ and $\mathcal{S}_{12\times2}$, the first $2n_t$ complex symbols ($n_t =3$ and $6$ respectively for $\mathcal{S}_{6\times2}$ and $\mathcal{S}_{12\times2}$) corresponding to the diagonal block matrices are conditionally two-group decodable for each of the possibilities for the remaining $2n_t$ symbols. Hence, the ML-decoding complexity of $\mathcal{S}_{6\times2}$ is of the order of $(M^6)(M^{2.5}) = M^{8.5}$ and that of $\mathcal{S}_{12\times2}$ is of the order of $(M^{12})(M^{5.5}) = M^{17.5}$ (the reduction by a factor of $\sqrt{M}$ is again due to the usage of hard-limiting). 

\begin{figure}
\centering
\includegraphics[width=3in,height=2.6in]{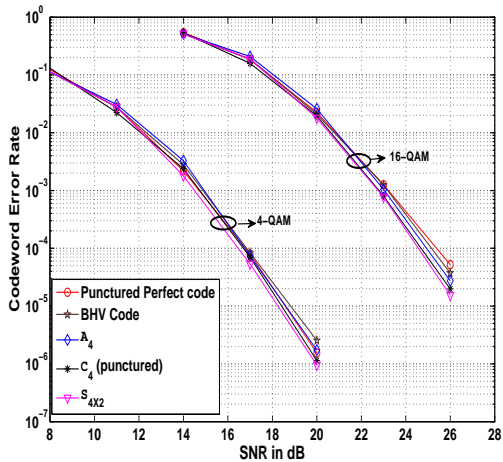}
\caption{CER performance of various rate-$2$ STBCs for the $4\times2$ system with $4$-$/16$-QAM}
\label{fig2}
\end{figure}

\begin{figure}
\centering
\includegraphics[width=3in,height=2.6in]{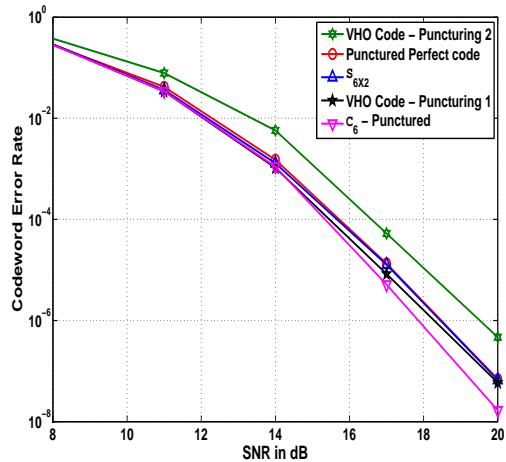}
\caption{CER performance of some well-known rate-$2$ STBCs for the $6\times2$ system at 4 bits per channel use}
\label{fig3}
\end{figure}

\section{Comparison with existing STBCs}\label{sec_sim}
We compare the performance of the STBCs constructed in this paper with some of the best known STBCs.
\subsection{$4\times2$ MIDO system}
As rival codes for $\mathcal{S}_{4\times2}$, we consider the following four STBCs - the punctured perfect code for $4$ transmit antennas (two of its layers have zero entries), the BHV code \cite{BHV}, the rate-$2$ STBC called $A_4$ code which is obtained in \cite[Section VIII-A]{roope}, and a new STBC obtained by puncturing $\mathcal{C}_4$ \cite{pav_modified}. $A_4$ has been shown \cite{roope} to be the best performing code among all the linear STBCs proposed in \cite{roope} for the $4\times 2$ MIDO system. The fourth rival code is obtained from $\mathcal{C}_4$ by simply puncturing the symbols corresponding to the basis elements $\zeta_5^2$ and $\zeta_5^3$, i.e., the entries of the first column of the codeword matrices are of the form $s_{i1} + s_{i2}\zeta_5$. This STBC has the best coding gain which can be explicitly calculated and is shown in Table \ref{tab1}. Even though the BHV code is not a full-diversity STBC, it is considered here since it is the first fast-decodable STBC proposed for the $4\times 2$ system, having an ML-decoding complexity of $\mathcal{O}(M^{4.5})$ for square $M$-QAM. We have not considered the other full-diversity STBCs proposed in \cite{FVC} - \cite{markin} since these codes have not been constructed with a focus on coding gain but only with an intention of having fast-decodability with a proven NVD. The constellations used in our simulations are $4$-QAM and $16$-QAM. 

Fig. \ref{fig2} reveals that $\mathcal{S}_{4\times2}$ has the best error performance among all codes under comparison, although punctured $\mathcal{C}_4$ has the best coding gain. This can possibly be attributed to the multiplicity of the minimum determinant - the number of codeword-difference matrices whose squared absolute value of determinant is the actual minimum determinant. We believe that punctured $\mathcal{C}_4$ has more such pairs since it is obtained by puncturing $\mathcal{C}_4$ and our method of puncturing might not be efficient. Punctured $\mathcal{C}_4$ loses only slightly to $\mathcal{S}_{4\times2}$ and these two codes beat the other three STBCs for both $4$- and $16$-QAM. 

\subsection{$6\times2$ MIDO system}
For this system, the rival codes for $\mathcal{S}_{6\times2}$ are the punctured perfect code for $6$ antennas \cite{ORBV} ($4$ layers punctured), punctured $\mathcal{C}_6$ \cite{pav_modified}, and two versions of the VHO-code for 6 transmit antennas \cite[Section X-C]{roope}. $\mathcal{C}_6$ is obtained from the CDA $(\mathbb{Q}(\omega,\zeta_7)/\mathbb{Q}(\omega), \tau:\zeta_7 \mapsto \zeta_7^3, -\omega)$, where $\zeta_7$ is the primitive $7^{th}$ root of unity. The entries of the first column of the codeword matrices of the punctured $\mathcal{C}_6$ are of the form $s_{i1} + s_{i2}\zeta_7$, $s_{ij} \in M\textrm{-HEX}$. The VHO-code for $6$ transmit antennas is a rate-3 STBC obtained from the CDA $\mathbb{Q}(\zeta_7/\mathbb{Q},\sigma:\zeta_7 \mapsto \zeta_7^3, -3/4)$. The first version of the VHO-code for the $6\times2$ system is obtained by puncturing the rate-$3$ VHO-code to obtain a rate-$2$ STBC, with the method of puncturing as depicted in \cite[Section X-C]{roope}. This STBC has an ML-decoding complexity of $\mathcal{O}\left(M^{8.5}\right)$. The second version of the rate-$2$ VHO-code is obtained by using the $\mathbb{Q}$-basis of $\mathbb{Q}(\zeta_7)$ to be $\{\zeta_7+\zeta_7^6,\zeta_7-\zeta_7^6,\zeta_7^2+\zeta_7^5,\frac{\zeta_7^2-\zeta_7^5}{2},\frac{\zeta_7^3+\zeta_7^4}{2},\zeta_7^3-\zeta_7^4 \}$ instead of the integral basis $\{1, \zeta_7,\zeta_7^2,\zeta_7^3,\zeta_7^4,\zeta_7^5\}$. This change of basis results in the overall ML-decoding complexity being only $\mathcal{O}(M^7)$ (Note that all the ML-decoding complexities are calculated after employing hard-limiting). Both versions of the VHO-code use $4$-QAM while the other STBCs use $4$-HEX constellation. 

Fig \ref{fig3} shows that $\mathcal{S}_{6\times2}$, the punctured perfect code, and the first version of the VHO-code (marked as ``VHO-code -Puncturing 1'' in the figure) have a very similar error performance. The second version of the VHO-code has poorer error performance but lower ML-decoding complexity. The best performance is that of punctured $\mathcal{C}_6$ which has the largest normalized minimum determinant.

\section{Discussion}\label{sec_discussion}
In this paper, we proposed a new method to obtain full-diversity, rate-$2$ STBCs from nonassociative algebras. We then constructed rate-$2$, fast-decodable STBCs for $4\times2$, $6 \times 2$, $8\times2$ and $12 \times 2$ systems which have large normalized minimum determinants, and STBC-schemes consisting of these STBCs have a non-vanishing determinant (NVD) so that they are DMT-optimal for their respective MIDO systems. We also showed that the Srinath-Rajan code has the same algebraic structure as the STBC constructed in this paper for the $4\times2$ system, thereby proving a previous conjecture that the STBC-scheme based on the Srinath-Rajan code has the NVD property and hence is DMT-optimal for the $4\times2$ system. However, there is still scope for improvement. Firstly, with the exception of the STBC for $4\times2$ MIDO system, the remaining STBCs in this paper have a lot of zero entries and naturally, there is the issue of high peak to average power ratio (PAPR) which needs to be lowered. Secondly, it is natural to seek conditions that enable the construction of higher rate codes (rate $>2$) with high coding gain and fast-decodability on the lines of the STBCs constructed in this paper. These are the possible directions for future research.

\appendices 

\section{Proof of Theorem \ref{thm1}}\label{app_1}
From Lemma \ref{lem1}, we know that $A_0 + \textbf{i}A_1$, when it has a right inverse, has a unique right inverse if and only if it is not a left zero divisor. To prove the theorem, we first show that any element $A_0 + \textbf{i}A_1 \in \mathcal{M}_{\mathcal{A}}$ is a left zero divisor if and only if equality holds in \eqref{condt1}. Following this, we show the existence of the right inverse to complete the proof of the theorem.

Suppose that $A_0 + \textbf{i}A_1$ is a left zero divisor of an element $B_0 + \textbf{i}B_1 + \cdots + \textbf{i}^{n-1}B_{n-1}$, $B_i \in \mathcal{A}$. Since $A_0 $ and $A_1$ are from $\mathcal{A}$ which is a CDA, we can assume that neither of $A_0$ and $A_1$ is zero since otherwise the unique right inverse always exists. Now, from our assumption,
\begin{eqnarray*}
 (A_0 + \textbf{i}A_1)(B_0 + \textbf{i}B_1 + \cdots + \textbf{i}^{n-1}B_{n-1}) &= & 0.
\end{eqnarray*}
So, noting that $A_1$ is invertible with its right inverse denoted by $A_1^{-1}$ (also its left inverse as elements of a CDA have the same left and right inverses), we have
\begin{eqnarray*}
 (A_0 ^\prime + \textbf{i})(B_0^\prime + \textbf{i}B_1^\prime + \cdots + \textbf{i}^{n-1}B_{n-1}^\prime) &= & 0
\end{eqnarray*}
where $A_0^\prime = A_0A_1^{-1}$, $B_i^\prime = \Upsilon^i(A_1)B_i $, $i=0,1,\cdots,n-1$. Due to the linear independence of $1$, $\textbf{i}$, $\cdots$, $\textbf{i}^{n-1}$ over $\mathcal{A}$, we have
\begin{eqnarray}
\label{t1e1}
 A_0^\prime B_0^\prime + \gamma_{{}_M}B_{n-1}^\prime & = &0,\\
\label{t1e2}
B_{k-1}^\prime + \Upsilon^k(A_0^\prime)B_k^\prime &  = &0, ~~k= 1,\cdots,n-1.
\end{eqnarray}
From \eqref{t1e1}, \eqref{t1e2} and the fact that $A_0^\prime \neq 0$, it is clear that $B_i^\prime \neq 0$, $i = 0,1,\cdots, n-1$. Solving \eqref{t1e2}, we arrive at 
$B_{0}^\prime  = (-1)^{n-1}\Upsilon(A_0^\prime )\Upsilon^{2}(A_0^\prime )\cdots\Upsilon^{n-1}(A_0^\prime )B_{n-1}^\prime $ using which in \eqref{t1e1}, we obtain
\begin{equation*}
 \left[(-1)^{n-1}A_0^\prime \Upsilon(A_0^\prime )\Upsilon^{2}(A_0^\prime )\cdots\Upsilon^{n-1}(A_0^\prime ) + \gamma_{{}_M}\right]B_{n-1}^\prime  = 0.
\end{equation*}
	Since $B_{n-1}^\prime  \neq 0$, we have $A_0^\prime \Upsilon(A_0^\prime)\Upsilon^{2}(A_0^\prime )\cdots\Upsilon^{n-1}(A_0^\prime ) = (-1)^{n}\gamma_{{}_M}$.  Taking $-A_0^\prime = C$, we have 
\begin{equation}\label{t1e3}
 C \Upsilon(C)\Upsilon^{2}(C )\cdots\Upsilon^{n-1}(C) = \gamma_{{}_M}.
\end{equation}
So, elements of the form $A_0 + \textbf{i}A_1$ are left zero divisors if and only if \eqref{t1e3} is satisfied. Therefore, if no $C \in \mathcal{A}$ satisfies \eqref{t1e3}, any element of the form $A_0 + \textbf{i}A_1$ has a unique right inverse which can be computed by equating the left hand side of \eqref{t1e1} with $1$. The resulting right inverse is obtained to be $B = B_0 + \textbf{i}B_1 + \cdots + \textbf{i}^{n-1}B_{n-1}$ where 
\begin{eqnarray}
\label{t1e4}
B_i & =  &\left[\Upsilon^i(A_1)\right]^{-1}B_i^\prime, ~i=0,1,\cdots,n-1,\\
\label{t1e5}
 B_{n-1}^\prime & = &\left[(-1)^{n-1}\prod_{i=0}^{n-1}\Upsilon^i(A_0^\prime)  + \gamma_{{}_M}\right]^{-1},\\
\label{t1e6}
B_{n-k}^\prime & = & (-1)^{k-1}\left(\prod_{i=n-k+1}^{n-1}\Upsilon^{i}(A_0^\prime)\right)B_{n-1}^\prime
\end{eqnarray}
\noindent and $A_0^\prime = A_0A_1^{-1}$. This completes the proof of Theorem \ref{thm1}.

\section{Proof of Lemma \ref{lem2}}\label{app_2}
For convenience, we denote $\textbf{C}\Upsilon(\textbf{C})\Upsilon^2(\textbf{C})\cdots\Upsilon^{n-1}(\textbf{C})$ by $\textbf{N}_C \in K^{m \times m}$. We first note that $\textbf{N}_C - \gamma_{{}_M}\textbf{I}$ is not invertible if and only if $\gamma_{{}_M}$ is an eigenvalue of $\textbf{N}_C$. This is because if $\gamma_{{}_M}$ is indeed an eigenvalue of $\textbf{N}_C$, then $\textbf{N}_C\textbf{x} = \gamma_{{}_M}\textbf{x}$ so that $\textbf{N}_C - \gamma_{{}_M}\textbf{I}$ is not full-ranked. Conversely, if $\textbf{N}_C - \gamma_{{}_M}\textbf{I}$ is not full-ranked, we have $\gamma_{{}_M}$ to be one of its eigenvalues. We now proceed to prove that $\gamma_{{}_M}$ is not an eigenvalue of $\textbf{N}_C$ when $C\Upsilon(C)\Upsilon^2(C)\cdots\Upsilon^{n-1}(C) \triangleq N_C \neq \gamma_{{}_M}$ for any $C \in \mathcal{A}$.

Suppose that $\gamma_{{}_M}$ is an eigenvalue of $\textbf{N}_C$. We first establish that the eigenvector of $\textbf{N}_C$ associated with $\gamma_{{}_M}$ has entries\footnote{In general, for any square matrix with entries from a field $K$, its eigenvalues and the entries of the associated eigenvectors need not be in $K$ but will be in the algebraic closure of $K$.} in $K$. Since $\gamma_{{}_M}$ is an element of the maximal subfield $K$, the entries of the rank-deficient matrix $\textbf{N}_C - \gamma_{{}_M}\textbf{I}$ are all elements of $K$. Hence, $\textbf{N}_C - \gamma_{{}_M}\textbf{I}$ can be viewed as the matrix of a linear transformation from the $m$-dimensional vector space $K^{m\times1}$ (over $K$) to itself with the kernel of the transformation being nontrivial and consisting of the eigenvectors of $\textbf{N}_C$ associated with $\gamma_{{}_M}$. We choose one such eigenvector and denote it by $\textbf{e}$. So, we have 
\begin{equation}\label{l2e1}
 \textbf{N}_C\textbf{e} = \gamma_{{}_M}\textbf{e}.
\end{equation}
Now, we note that $\textbf{N}_C$ is also obtained by left regular representation \cite{jacobson} as the matrix of the linear transformation $ \lambda_{N_C} : \mathcal{A} \rightarrow \mathcal{A}$, with $\lambda_{N_C}(B) = N_cB$, $\forall B \in \mathcal{A}$. Observing that any element of $\mathcal{A}$ can be expressed as $[ 1,~\textbf{j},~\cdots,~\textbf{j}^{m-1}]\textbf{k}$, $\textbf{k} \in K^{m\times1}$, let $E = [ 1,~\textbf{j},~\cdots,~\textbf{j}^{m-1}]\textbf{e}$, with $\textbf{e}$ defined in \eqref{l2e1}. So,
\begin{eqnarray}
\label{l2e2}
 \lambda_{N_C}(E) = N_cE & = &[ 1,~\textbf{j},~\cdots,~\textbf{j}^{m-1}]\textbf{N}_C\textbf{e} \\
\label{l2e3}
& = & [ 1,~\textbf{j},~\cdots,~\textbf{j}^{m-1}]\gamma_{{}_M}\textbf{e}\\
\label{l2e4}
& = & E\gamma_{{}_M},
\end{eqnarray}
where \eqref{l2e2} is by definition of left regular representation, \eqref{l2e3} is due to \eqref{l2e1}, and \eqref{l2e4} follows by noting that $\gamma_{{}_M}$ is an element of $K$. Hence, $E^{-1}N_CE = \gamma_{{}_M}$. Since we have denoted $\textbf{C}\Upsilon(\textbf{C})\Upsilon^2(\textbf{C})\cdots\Upsilon^{n-1}(\textbf{C})$ by $\textbf{N}_C$, we have
\begin{eqnarray}
 \nonumber
  \gamma_{{}_M} & = & E^{-1}C\Upsilon(C)\Upsilon^2(C)\cdots\Upsilon^{n-1}(C)E\\ \nonumber
& = & E^{-1}C\left(\Upsilon(E)\left(\Upsilon(E)\right)^{-1}\right)\Upsilon(C) \\ \nonumber
&& \times \left(\Upsilon^2(E)\left(\Upsilon^2(E)\right)^{-1}\right)\Upsilon^2(C)\\ \nonumber
&& \times \cdots \left(\Upsilon^{n-1}(E)\left(\Upsilon^{n-1}(E)\right)^{-1}\right)\Upsilon^{n-1}(C)E \\
\label{l2e5}
 & = & C^\prime \Upsilon(C^\prime)\Upsilon^2(C^\prime)\cdots\Upsilon^{n-1}(C^\prime)
\end{eqnarray}
where $C^\prime \triangleq E^{-1}C\Upsilon(E)$ and \eqref{l2e5} is obtained using \eqref{d1} and \eqref{d2} and also noting that $\left(\Upsilon(E)\right)^{-1} = \Upsilon(E^{-1})$ (since $\Upsilon(E)\Upsilon(E^{-1}) = \Upsilon(EE^{-1}) = 1$). But \eqref{l2e5} leads to a contradiction since there exists no $C \in \mathcal{A}$ such that $C\Upsilon(C)\Upsilon^2(C)\cdots\Upsilon^{n-1}(C) = \gamma_{{}_M}$. Therefore, $\gamma_{{}_M}$ is never an eigenvalue of $\textbf{C}\Upsilon(\textbf{C})\Upsilon^2(\textbf{C})\cdots\Upsilon^{n-1}(\textbf{C})$ which proves Lemma \ref{lem2}. 

\section{Proof of Theorem \ref{Thm2}}\label{app_3}

Let $\mathcal{A}_{mat}$ be the ring of $m\times m$ sized invertible matrices that are representations of elements of $\mathcal{A}$, i.e.,
\begin{equation*}
 \mathcal{A}_{mat} = \{\textbf{A} \vert \textbf{A} \textrm{ is the matrix representation of }A \in \mathcal{A} \}.
\end{equation*}
 We have already assumed that $\gamma_{{}_M}$ does not belong to the centre of $\mathcal{A}$ and so does not commute with every element of $\mathcal{A}$ because of which $\mathcal{M}_\mathcal{A}$ does not have a matrix representation in $M_{mn}(K)$. But every element of $\mathcal{A}_{mat}$ commutes with $\gamma_{{}_M}\textbf{I}$ (which is not the matrix representation of $\gamma_{{}_M}$ in $M_{m}(K)$ and does not belong to $\mathcal{A}_{mat}$). For any finite or infinite set $\mathcal{P}$ of $m \times m$ sized matrices, we use the notation $\mathcal{A}_{mat}[\mathcal{P}]$ to denote the ring of $m\times m$ matrices generated by $\mathcal{P}$ over $\mathcal{A}_{mat}$. We now consider the ring $\mathcal{A}_{mat}[\gamma_{{}_M}\textbf{I}]$ which is not a division ring (for example, if ${\bf \Gamma}_M$ denotes the matrix representation of $\gamma_{{}_M}$ in $M_{m}(K)$, then ${\bf \Gamma}_M-\gamma_{{}_M}\textbf{I}$ is not invertible). Let  $\mathcal{A}_{inv-mat}[\gamma_{{}_M}\textbf{I}] = \{\textbf{B} \vert \textbf{B}^{-1} \in \mathcal{A}_{mat}[\gamma_{{}_M}\textbf{I}] \}$, i.e., the set of inverses of all invertible matrices in $\mathcal{A}_{mat}[\gamma_{{}_M}\textbf{I}]$. Next, consider the infinite ring $\mathcal{M}$ whose elements are matrices of the form 
\begin{equation} \label{inv_M}
 \textbf{B} = \left[ \begin{array}{cccc}
                      \textbf{B}_0 & \gamma_{{}_M}\Upsilon(\textbf{B}_{n-1}) & \cdots & \gamma_{{}_M}\Upsilon^{n-1}(\textbf{B}_{1})\\
                      \textbf{B}_1 & \Upsilon(\textbf{B}_0) & \cdots & \gamma_{{}_M}\Upsilon^{n-1}(\textbf{B}_{2})\\
                      \textbf{B}_2 & \Upsilon(\textbf{B}_1) & \cdots & \gamma_{{}_M}\Upsilon^{n-1}(\textbf{B}_{3})\\
                      \textbf{B}_3 & \Upsilon(\textbf{B}_{2}) & \cdots & \gamma_{{}_M}\Upsilon^{n-1}(\textbf{B}_{4})\\
                      \vdots & \vdots & \ddots  & \vdots\\
                      \textbf{B}_{n-1} & \Upsilon(\textbf{B}_{n-2}) & \cdots & \Upsilon^{n-1}(\textbf{B}_0)\\                      
                     \end{array} \right]
\end{equation}
with $\textbf{B}_i \in \mathcal{A}_{mat}[\{\gamma_{{}_M}\textbf{I}\} \cup \mathcal{A}_{inv-mat}[\gamma_{{}_M}\textbf{I}]]$, $i=0,\cdots,n-1$. With these facts developed, we proceed with the proof of the theorem as follows.

{\bf 1)} {\it Proof that $\textbf{M}$ is invertible}: Let $B = B_0 +\textbf{i}B_1+\cdots+ \textbf{i}^{n-1}B_{n-1}$ be the unique right inverse of $A_0+\textbf{i}A_1$ given by \eqref{t1e4} - \eqref{t1e6} with $A_0,A_1,B_0,\cdots,B_{n-1} \in \mathcal{A}$. Let
\begin{eqnarray}
\label{l3e1}
 \textbf{A}_0^\prime & \triangleq &\textbf{A}_0\textbf{A}_1^{-1}, \\ 
\label{l3e2}
 \textbf{B}_{n-1}^\prime & \triangleq &\left[(-1)^{n-1}\prod_{i=0}^{n-1}\Upsilon^{i}(\textbf{A}_0^\prime ) + \gamma_{{}_M}\textbf{I}\right]^{-1},
\end{eqnarray}
\begin{eqnarray}
\label{l3e3}
\textbf{B}_{n-k}^\prime & \triangleq  & (-1)^{k-1}\left(\prod_{i=n-k+1}^{n-1}\Upsilon^{i}(\textbf{A}_0^\prime)\right)\textbf{B}_{n-1}^\prime,\\
\label{l3e4}
\textbf{B}_i & \triangleq  &\left[\Upsilon^i(\textbf{A}_1)\right]^{-1}\textbf{B}_i^\prime, ~i=0,1,\cdots,n-1.	
\end{eqnarray}
\noindent The existence of $\textbf{B}_{n-1}^\prime$ can be verified by applying Theorem \ref{thm1} and Lemma \ref{lem2} in that order. The inverse of $\textbf{M}$ has the form shown in \eqref{inv_M} with $\textbf{B}_i$ obtained using \eqref{l3e1}-\eqref{l3e4}. To check that this matrix, denoted by $\textbf{M}_{inv}$, is indeed the inverse of $\textbf{M}$, note that both $\textbf{M}$ and $\textbf{M}_{inv}$ belong to $\mathcal{M}$ and hence their product also is in $\mathcal{M}$. So, it only suffices to check that the first $m$ columns of the product of $\textbf{M}$ and $\textbf{M}_{inv}$ are $[\textbf{I},\textbf{O},\textbf{O},\cdots,\textbf{O}]^T$, which follows upon using \eqref{l3e1}-\eqref{l3e4}.

{\bf 2)} {\it Proof that $det(\textbf{M}) \in L$}: It can be noted that $\textbf{M} \in K^{nm \times nm}$ so that $det(\textbf{M}) \in K$. Also, $\tau^{i}\left(det(\textbf{M})\right) = det\left(\Upsilon^{i}(\textbf{M})\right)$ where, as mentioned before, $\Upsilon^{i}(\textbf{M})$ refers to the matrix obtained by applying $\tau^{i}$ to each entry of $\textbf{M}$, $i=0,1,\cdots,n-1$. To prove that $det(\textbf{M}) \in L$, it suffices to show that $det\left(\Upsilon(\textbf{M})\right) = det(\textbf{M})$ since the only elements fixed by $Gal(K/L) = \langle \tau \rangle$ are the elements of $L$. Let $\textbf{P}(i,j)$ denote the $(i,j)^{th}$ entry of a matrix $\textbf{P}$. Consider permutation matrices $\textbf{P}_1$ and $\textbf{P}_2$ whose nonzero elements are 
\begin{eqnarray*}
 \textbf{P}_1(k, (n-1)m+k) & = & 1,  ~~ k = 1,2,\cdots, m, \\
 \textbf{P}_1(k, k-m) & = & 1,  ~~ k = m+1,m+2,\cdots, nm,\\
 \textbf{P}_2(k, m+k) & = & 1,  ~~ k = 1,2,\cdots, (n-1)m, \\
 \textbf{P}_2(k, k-(n-1)m) & = & 1,  ~~ k = (n-1)m+1,\cdots, nm.
\end{eqnarray*}
Now, $\textbf{P}_1\Upsilon(\textbf{M})\textbf{P}_2$ has the following structure.
\begin{equation} \label{M_form}
 \left[ \begin{array}{cccccc}
                      \textbf{A}_0 & \textbf{O} & \cdots & \textbf{O} & \cdots & \Upsilon^{n-1}(\textbf{A}_1)\\
                      \gamma_{{}_M}\textbf{A}_1 & \Upsilon(\textbf{A}_0) & \cdots & \textbf{O} & \cdots & \textbf{O}\\
                      \textbf{O} & \Upsilon(\textbf{A}_1) & \cdots & \vdots & \cdots & \textbf{O}\\
                      \textbf{O} & \textbf{O} & \cdots & \Upsilon^{i-1}(\textbf{A}_0) & \cdots & \textbf{O}\\
                      \vdots & \vdots & \cdots &\Upsilon^{i-1}(\textbf{A}_1) & \cdots & \textbf{O}\\
                      \vdots & \vdots & \vdots & \cdots & \cdots & \textbf{O}\\
                      \textbf{O} & \textbf{O} & \cdots & \textbf{O} & \cdots & \Upsilon^{n-1}(\textbf{A}_0)\\                      
                     \end{array} \right].
\end{equation}
Therefore, with diagonal matrices $\textbf{G}_1$ and $\textbf{G}_2$ whose nonzero diagonal elements are defined as 
\begin{eqnarray*}
 \textbf{G}_1(k,k)&  = & \gamma_{{}M}, ~~ k = 1,2,\cdots,m, \\
 \textbf{G}_1(k,k)&  = & 1, ~~ k = m+1,\cdots,nm,\\
\textbf{G}_2(k,k)&  = & \gamma_{{}M}^{-1}, ~~ k = 1,2,\cdots,m,\\
\textbf{G}_2(k,k)&  = & 1, ~~ k = m+1,\cdots,nm,
\end{eqnarray*}
we observe that $\textbf{M} = \textbf{G}_1 \textbf{P}_1\Upsilon(\textbf{M})\textbf{P}_2 \textbf{G}_2$ so that $det(\textbf{M}) = det\left(\Upsilon(\textbf{M})\right)$ (for $det(\textbf{G}_1)det(\textbf{G}_2) = 1$ and $\textbf{P}_1$ and $\textbf{P}_2$ are permutation matrices). Therefore $det(\textbf{M}) \in L$. 

\section{Proof of Proposition \ref{lem4}}\label{app_4a}
\noindent Let $A = a +\textbf{j}b$, $a,b \in \mathbb{Q}(i,\sqrt{5})$. Suppose that 
\begin{equation}\label{l4e1}
 A\Upsilon(A) = i.
\end{equation}
 Now, if $b = 0$, then $a\tau(a) = i$ which is not a possibility in $(\mathbb{Q}(i,\sqrt{5})/\mathbb{Q}(i),\tau,i)$ (which has $i$ as its non-norm element). If $ a= 0$, we have $\textbf{j}b\textbf{j}\tau(b) = i$ so that 
\begin{equation}\label{l4e2}
 \sigma(b)\tau(b) = -i.
\end{equation}
Applying $\sigma$ throughout in \eqref{l4e2}, we get $b\sigma\tau(b) = i$. Next, applying $\tau$ throughout in \eqref{l4e2}, we get $b\sigma\tau(b) = -i$ which leads to a contradiction. So, \eqref{l4e2} is not true (Note that $\tau^2$ is identity and $\tau\sigma = \sigma\tau$) and we can assume that $a, b \neq 0$. Now, applying $\Upsilon$ throughout in \eqref{l4e1}, we obtain  $\Upsilon(A)A = i$. Hence,
\begin{equation*}
 (a+ \textbf{j}b)(\tau(a) + \textbf{j}\tau(b)) = (\tau(a) + \textbf{j}\tau(b))(a+ \textbf{j}b)
\end{equation*}
which leads to 
\begin{equation}\label{l4e3}
 \frac{\sigma(b)}{\sigma\tau(b)} = \sigma\left(\frac{b}{\tau(b)}\right) = \frac{b}{\tau(b)}
\end{equation}
Hence, $b/ \tau(b) $ is invariant under $\sigma$ and hence belongs to $\mathbb{Q}(\sqrt{5})$. Also, from \eqref{l4e1}, we have
\begin{eqnarray}\label{l4e4}
 a\tau(a) - \sigma(b)\tau(b) & = & i\\
\nonumber
b\tau(a) + \sigma(a)\tau(b) & = & 0
\end{eqnarray}
so that 
\begin{equation}\label{l4e6}
 \frac{b}{\tau(b)} = - \frac{\sigma(a)}{\tau(a)}.
\end{equation}
Using \eqref{l4e6} in \eqref{l4e4}, we obtain
\begin{equation*}
 \frac{\tau(a)}{\sigma(a)}\left( a\sigma(a) + b\sigma(b) \right) = i.
\end{equation*}
Now, $a\sigma(a) + b\sigma(b)$ is invariant under $\sigma$ and hence is in $\mathbb{Q}(\sqrt{5})$. So, $\tau(a)/\sigma(a)$ is imaginary and belongs to $\mathbb{Q}(i,\sqrt{5})$ using which in \eqref{l4e6}, we note that $b/ \tau(b)$ is also imaginary. This contradicts the earlier result obtained below \eqref{l4e3}. Therefore, our assumption that $A\Upsilon(A) = i$ is false which proves the lemma.

\section{Proof of Proposition \ref{prop_a} }\label{app_4}  

\noindent Let $A = a + \textbf{j}b$, $a,b \in \mathbb{Q}(\omega,\theta)$ such that 
\begin{equation}\label{a1e1}
 A\Upsilon(A)\Upsilon^2(A) = \omega.
\end{equation}
Firstly, $a \neq 0$ since otherwise $(\textbf{j}b)(\textbf{j}\tau(b))(\textbf{j}\tau^2(b)) = \omega $ which is not possible. Secondly, $b \neq 0$ since $a\tau(a)\tau^2(a) \neq \omega$ for any $a \in \mathbb{Q}(\omega,\theta)$ (for $\omega$ is a non-norm element in $(\mathbb{Q}(\omega,\theta)/\mathbb{Q}(\omega),\tau, \omega)$). Hence, we assume that $a,b\neq0$. Applying $\Upsilon^2$ throughout in $\eqref{a1e1}$, we obtain $\Upsilon^2(A)A\Upsilon(A) = \omega$ so that 
$A\Upsilon(A)\Upsilon^2(A) = \Upsilon^2(A)A\Upsilon(A)$. Now, $A\Upsilon(A) = x + \textbf{j}\sigma(y)$ where $x = a \tau(a) - \sigma(b)\tau(b)$, $\sigma(y) = b\tau(a) + \sigma(a)\tau(b)$. So, 
\begin{equation*}
 \left(x+\textbf{j}\sigma(y)\right)(\tau^2(a)+\textbf{j}\tau^2(b)) = (\tau^2(a)+\textbf{j}\tau^2(b))\left(x+\textbf{j}\sigma(y)\right)
\end{equation*}
from which we have
\begin{equation}\label{a1e2}
 \sigma(y)\tau^2(a) + \sigma(x)\tau^2(b) = \sigma\tau^{2}(a)\sigma(y)+x\tau^2(b).
\end{equation}
From \eqref{a1e1}, we obtain
\begin{eqnarray}
 \label{a1e3}
x\tau^2(a) - y\tau^2(b) &  = & \omega \\
\label{a1e4}
\sigma(x)\tau^2(b) + \sigma(y)\tau^2(a) & = &0.
\end{eqnarray}
If $x = 0$, then $y =0$ (since $a \neq 0$) and \eqref{a1e3} is not true. So, we can assume $x \neq 0$. Using \eqref{a1e2} and \eqref{a1e4}, we get 
\begin{equation*}\label{a1e5}
\frac{\sigma(y)}{\tau^2(b)} = -\frac{x}{\sigma\tau^2(a)} = -\frac{\sigma(x)}{\tau^2(a)} 
\end{equation*}
so that $\frac{x}{\sigma\tau^2(a)} = \sigma\left(\frac{x}{\sigma\tau^2(a)}\right)$. Therefore, $x/\sigma\tau^2(a)$ is real-valued and belongs to $\mathbb{Q}(\theta)$. Now, using \eqref{a1e4} in \eqref{a1e3}, we get
\begin{equation*}
 \frac{\tau^2(a)}{\sigma(x)}\big[ x\sigma(x) + y\sigma(y) \big] = \omega
\end{equation*}
Since $x\sigma(x) + y\sigma(y)$ is invariant under $\sigma$ and hence real-valued, $\frac{\tau^2(a)}{\sigma(x)}$ must be complex-valued which contradicts the previous result. Hence, \eqref{a1e1} is false and there exists no $A \in \mathcal{A}$ such that $A\Upsilon(A)\Upsilon^2(A) = \omega$.

\section{Proof of Proposition \ref{prop_b} }\label{app_5}  

\noindent Let $A = a + \textbf{j}b$, $a,b \in \mathbb{Q}(i,\theta)$ such that 
\begin{equation}\label{a2e1}
 A\Upsilon(A)\Upsilon^2(A)\Upsilon^3(A) = i.
\end{equation}
 Applying $\Upsilon^2$ throughout \eqref{a2e1}, we obtain $\Upsilon^2(A)\Upsilon^3(A)A\Upsilon(A) = i$ so that $A\Upsilon(A)\Upsilon^2(A)\Upsilon^3(A) = \Upsilon^2(A)\Upsilon^3(A)A\Upsilon(A)$. Let $A\Upsilon(A) = x + \textbf{j}y$ where $x = a \tau(a) - \sigma(b)\tau(b)$, $y = b\tau(a) + \sigma(a)\tau(b)$. So, we have
\begin{equation*}
 \left(x+\textbf{j}y\right)(\tau^2(x)+\textbf{j}\tau^2(y)) = (\tau^2(x)+\textbf{j}\tau^2(y))\left(x+\textbf{j}y\right)
\end{equation*}
from which we obtain
\begin{equation}\label{a2e3}
 y\tau^2(x) + \sigma(x)\tau^2(y) = \sigma\tau^{2}(x)y+ x\tau^2(y).
\end{equation}
From \eqref{a2e1}, we obtain
\begin{eqnarray}
 \label{a2e4}
x\tau^2(x) - \sigma(y)\tau^2(y) &  = & i\\
\label{a2e5}
\sigma(x)\tau^2(y) + y\tau^2(x) & = &0.
\end{eqnarray}
Now, if $x = 0$, 
\begin{equation}\label{temp}
 \sigma(y)\tau^2(y) = -i.
\end{equation}
By applying $\sigma$ and $\tau^2$ separately throughout \eqref{temp}, we obtain $y\sigma\tau^2(y) = i$ and $y\sigma\tau^2(y) = -i$ which contradict each other. Hence, $x \neq 0$. On the other hand, if $y = 0$, then $x\tau^2(x) = i$ and applying $\tau$ we get $\tau(x)\tau^3(x) = i$ so that $x\tau^2(x)\tau(x)\tau^3(x) = (i)(i) = -1$. Since $x \in \mathbb{Q}(i,\theta)$, this implies that there exists some $x \in \mathbb{Q}(i,\theta)$ such that $N_{\mathbb{Q}(i,\theta)/\mathbb{Q}(i)}(x) = -1$. This is not true since $i^t$ is not a norm of any field element of $\mathbb{Q}(i,\theta)$ for $t=1,2,3$. Hence, we can assume that $x,y \neq 0$. Using \eqref{a2e3} and \eqref{a2e5}, we get 
\begin{equation*}
-\frac{\sigma\tau^2(x)}{x} = \frac{\tau^2(y)}{y} = -\frac{\tau^2(x)}{\sigma(x)} 
\end{equation*}
so that $\frac{\tau^2(x)}{\sigma(x)} = \sigma\left(\frac{\tau^2(x)}{\sigma(x)}\right)$. Therefore, $\frac{\tau^2(x)}{\sigma(x)}$ is real-valued and belongs to $\mathbb{Q}(\theta)$. Now, using \eqref{a2e5} in \eqref{a2e4}, we get
\begin{equation*}
 \frac{\tau^2(x)}{\sigma(x)}\big[ x\sigma(x) + y\sigma(y) \big] = i.
\end{equation*}
Since $x\sigma(x) + y\sigma(y)$ is invariant under $\sigma$ and hence real-valued, it must be that $\frac{\tau^2(x)}{\sigma(x)}$ is complex-valued which contradicts the previous result. Hence, \eqref{a2e1} is false and there exists no $A \in \mathcal{A}$ such that $A\Upsilon(A)\Upsilon^2(A)\Upsilon^3(A) = i$.

\section{Proof of Proposition \ref{prop_c} }\label{app_6} 

\noindent Let $A = a + \textbf{j}b$, $a,b \in \mathbb{Q}(\omega,\theta)$ such that 
\begin{equation}\label{a3e1}
 A\Upsilon(A)\Upsilon^2(A)\cdots\Upsilon^5(A) = -\omega.
\end{equation}
Applying $\Upsilon^3$ throughout \eqref{a3e1}, we observe that $A\Upsilon(A)\Upsilon^2(A)$ and $\Upsilon^3(A)\Upsilon^4(A)\Upsilon^5(A)$ commute. Let $A\Upsilon(A)\Upsilon^2(A) = x + \textbf{j}y$ where $x = x^\prime\tau^2(a) - \sigma(y^\prime)\tau^2(b)$, $y = y^\prime\tau^2(a)+\sigma(x^\prime)\tau^2(b)$ with $x^\prime = a \tau(a) - \sigma(b)\tau(b)$, $y^\prime = b\tau(a) + \sigma(a)\tau(b)$. So, we have
\begin{equation*}
 \left(x+\textbf{j}y\right)(\tau^3(x)+\textbf{j}\tau^3(y)) = (\tau^3(x)+\textbf{j}\tau^3(y))\left(x+\textbf{j}y\right)
\end{equation*}
from which we obtain
\begin{equation}\label{a3e3}
 y\tau^3(x) + \sigma(x)\tau^3(y) = \sigma\tau^{3}(x)y+ x\tau^3(y).
\end{equation}
From \eqref{a3e1}, we obtain
\begin{eqnarray}
 \label{a3e4}
x\tau^3(x) - \sigma(y)\tau^3(y) &  = & -\omega\\
\label{a3e5}
\sigma(x)\tau^3(y) + y\tau^3(x) & = &0.
\end{eqnarray}
Now, if $x = 0$, 
\begin{equation}\label{temp1}
 \sigma(y)\tau^3(y) = \omega.
\end{equation}
By applying $\sigma$ and $\tau^3$ separately throughout \eqref{temp1}, we obtain $y\sigma\tau^3(y) = \omega^2$ (for $\sigma(\omega) = \omega^2$) and $y\sigma\tau^3(y) = \omega$, which contradict each other. Hence, $x \neq 0$. On the other hand, if $y = 0$, then $x\tau^3(x) = -\omega$ and therefore, $\tau(x)\tau^4(x) = -\omega$, $\tau^2(x)\tau^5(x) = -\omega$. Using these results, we arrive at
\begin{equation}\label{add_eq}
\left(x\tau^3(x)\right)\left(\tau(x)\tau^4(x)\right)\left( \tau^2(x)\tau^5(x)\right) = (-\omega)^3 = -1. 
\end{equation}
  Since $x \in \mathbb{Q}(\omega,\theta)$, \eqref{add_eq} implies that there exists some $x \in \mathbb{Q}(\omega,\theta)$ such that $N_{\mathbb{Q}(\omega,\theta)/\mathbb{Q}(\omega)}(x) = -1$. This is not true since $(-\omega)^t$ is not a norm of any field element of $\mathbb{Q}(\omega,\theta)$ for $t=1,\cdots,5$. Hence, we can assume that $x,y \neq 0$. Using \eqref{a3e3} and \eqref{a3e5}, we get 
\begin{equation*}
-\frac{\sigma\tau^3(x)}{x} = \frac{\tau^3(y)}{y} = -\frac{\tau^3(x)}{\sigma(x)} 
\end{equation*}
so that $\frac{\tau^3(x)}{\sigma(x)} = \sigma\left(\frac{\tau^3(x)}{\sigma(x)}\right)$. Therefore, $\frac{\tau^3(x)}{\sigma(x)}$ is real-valued and belongs to $\mathbb{Q}(\theta)$. Now, using \eqref{a3e5} in \eqref{a3e4}, we get
\begin{equation*}
 \frac{\tau^3(x)}{\sigma(x)}\big[ x\sigma(x) + y\sigma(y) \big] = \omega.
\end{equation*}
Since $x\sigma(x) + y\sigma(y)$ is invariant under $\sigma$ and hence real-valued, $\frac{\tau^3(x)}{\sigma(x)}$ has to be complex-valued which contradicts the previous result. Hence, \eqref{a3e1} is false and there exists no $A \in \mathcal{A}$ such that $A\Upsilon(A)\Upsilon^2(A)\cdots \Upsilon^5(A) = -\omega$.

\section*{Acknowledgements}
We thank Dr. Nadya Markin for useful discussions on the topic. We also thank the anonymous reviewers for their constructive suggestions that have greatly helped in improving the quality of the manuscript.

\end{document}